\documentclass[a4paper,12pt]{article}
\usepackage[utf8]{inputenc}
\usepackage{graphicx}
\usepackage{setspace}
\usepackage{amsthm}
\usepackage{amssymb}
\usepackage{mathtools}
\usepackage{lmodern}
\usepackage{hyperref}
\usepackage{tkz-euclide}
\usepackage{mathabx}
\usepackage{enumerate}
\usepackage[shortlabels]{enumitem}
\usepackage{float}
\usepackage{tikz}
\usepackage{tikz-3dplot}
\usepackage{bm}
\usepackage{epigraph}
\usepackage{textcomp}
\usepackage[makeroom]{cancel}
\usepackage{tkz-euclide}
\usetkzobj{all}
\usepackage{mathtools}

\theoremstyle{definition}
\newtheorem{definition}{Definition}
\newtheorem{theorem}{Theorem}

\theoremstyle{definition}
\newtheorem{Lemma}{Lemma}
 
\theoremstyle{definition}
\newtheorem*{remark}{Remark}

\theoremstyle{definition}

\theoremstyle{definition}

\theoremstyle{definition}

\title{A system of axioms for Minkowski spacetime\footnote{We are grateful to John Burgess, Dino Calosi, Harold Hodes and Chris Wüthrich for discussion and comments on parts of this paper.}}
\author{Lorenzo Cocco \space \space \space \space \space \space Joshua Babic }
\date{}

\begin{document}

\maketitle

\begin{abstract}
We present an elementary system of axioms for the geometry of Minkowski spacetime. It strikes a balance between a simple and streamlined set of axioms and the attempt to give a direct formalization in first-order logic of the standard account of Minkowski spacetime in [Maudlin 2012] and [Malament, unpublished]. It is intended for future use in the formalization of physical theories in Minkowski spacetime. The choice of primitives is in the spirit of [Tarski 1959]: a predicate of betwenness and a four place predicate to compare the square of the relativistic intervals. Minkowski spacetime is described as a four dimensional `vector space' that can be decomposed everywhere into a spacelike hyperplane - which obeys the Euclidean axioms in [Tarski and Givant, 1999] - and an orthogonal timelike line. The length of other `vectors' are calculated according to Pythagoras' theorem.
We conclude with a \textit{Representation Theorem} relating models $\mathfrak{M}$ of our system $\mathcal{M}^1$ that satisfy second order continuity to the mathematical structure $\langle \mathbb{R}^{4}, \eta_{ab}\rangle$, called `Minkowski spacetime' in physics textbooks.   
\end{abstract}

\section{Introduction and motivation}

The aim of this paper is to provide an elementary system of axioms that characterizes the geometry of Minkowski spacetime. It will be pursued in the style of Tarski; that is, with a primitive predicate of betweenness and a quaternary predicate to compare the relativistic intervals between points.\newline A system of this sort is needed, first of all, for certain investigations on the foundations of relativity. One question that we believe deserves attention is that of the theoretical equivalence of two types of formulations of relativity. There are `dynamical' formulations of relativity, framed in terms of observers, coordinates systems and the like [Andréka, Németi et al. 2011]. We can contrast them with `geometric' formulations of the theory that eschew this apparatus, describing the intrinsic features of a manifold of spacetime points. \\ 
Robb [1914, 1936] was the first to provide an axiomatic description of the geometry of spacetime, and in particular of its causal structure. He was soon followed by Reichenbach [1924]. Their systems are not formalized and make use of some unnecessary set theory. An excellent set of axioms that is entirely elementary and in first order logic has been formulated by Goldblatt [1987] in terms of orthogonality. Many other axiomatizations have been proposed.\footnote{The systems of Mundy [1986a, 1986b] are notable examples. Mundy [1986a] is close to that of Robb[1936] and is based on lightlike connectibility. Mundy [1986b] is the most similar to ours, but requires five primitives: three primitive notions of betweenness, timelike, spacelike and lightlike betweenness, and two primitive notions of congruence, temporal and spatial congruence. Other systems worth mentioning are that of [Ax 1978] and [Schutz 1997], although they both heavily rely on set-theoretic machinery. [Ax 1978] is a `dynamical' system (in our terminology). It employs variables of two sorts: one ranging over particles and one ranging over signals. It construes segments as sets of `particles'.} \newline Unfortunately all of these systems are rather unwieldy to work with, when one attempts to extract physics from them.  Just to account for the description of spacetime along the lines of [Maudlin 2012] and [Malament, unpublished] requires several pages of definitions and derivations. On the other hand, our preferred standard of theory equivalence is a modification of one due to [Barrett \& Halvorson 2016].\footnote{They themselves modify an earlier proposal of [Quine 1975]. The book [Halvorson 2019] surveys several such notions of equivalence for scientific theories and argues that a plausible candidate should be intermediate in strength between mutual interpretability and bi-interpretability. In a future paper, we will propose an ulterior refinement of [Barrett \& Halvorson 2016][Quine 1975] [Spector 1958] and defend that it is the best criterion of equivalence. We need to allow for the translation of theories with different domains of discourse,  as in all the generalized notions of interpretation described in [Alscher, 2016; chap. 1] and several natural examples of reconstrual in mathematics [Halvorson 2019, pp. 143-145] } Any reasonable definition requires a `dictionary' between talk of coordinates and  spacetime notions. We also need to derive the translation of the axioms of [Andréka, Németi et al. 2011] from the geometric theories and \textit{vice versa}. In future work, we plan to describe such a translation and consider some of its philosophical implications. But we have found it more convenient to give first an equivalent but more manageable theory, with simpler extralogical primitives, to act as an intermediate.\footnote{The variables of the system of [Andréka, Németi et al. 2011] range over bodies, observers and real numbers. Their primitive predicates are those of the theory of real-closed fields and a primitive predicate Cooordinatization $obxyzt$ that applies to an observer, a body and four coordinates in the obvious circumstances. The calculus of segments seems to be needed to translate this talk of localization relative to coordinates.} \newline 
A second, intrinsic justification for our system is that it allows a straightforward proof of the \textit{Representation Theorems} of [Tarski 1959] and [Tarski and Sczerba, 1979] for Minkowski Spacetime.\footnote{The price to pay is that our axioms cannot be stated simply in primitive notation. Tarski and Givant [1999, p.192/f], and most logicians working on geometry, attach much importance to avoiding defined symbols. This does not appear to us to be a decisive defect. In axiomatic set theory, nobody would take the pains to write down the axiom `\textit{V=L}', or \textit{Martin's axiom}, or the \textit{Proper Forcing Axiom} only in terms of quantifiers, truth functions and the epsilons. This does not disqualify them as possible additions to ZFC [Jensen 1972].} We use the results of [Tarski, 1959] for Euclidean space to show that: (1) every model of a [second order version] of our theory admits of a coordinatization into $\mathbb{R}^4$ and (2) any two such coordinatizations $f$ and $f'$ are equivalent up to rescaling $U$ and a Poincaré transformation $L$ (sec. VI). In addition, it is plausible that the system below can be be more easily supplemented to axiomatize a field theory, for example  electrodynamics.\footnote{Consider the problem of formalizing Maxwell's theory on the systems of Goldblatt [1987] and Mundy [1986a, 1986b]. To formulate a nominalistic analog to a system of partial differential equations - in the style of [Field 1980] - and set up an initial value problem, we are forced to introduce by definition the apparatus to describe a foliation and employ it in the axioms. This means that the main advantage that the systems of Goldblatt [1987] and Mundy [1986a] have over ours, the fact that they can be stated elegantly without abbreviative definitions, disappears when we come to relativistic electrodynamics.} It acts as a useful `buffer' between `dynamical' and geometric formulations of the theory. A proof of the equivalence of our system to our target system - in the sense of W.V.O Quine [1975] -  will \textit{ipso facto} carry over to other geometric systems of axioms that are interderivable.        
It will be evident enough how to derive from our system all the axioms in the appendix to Goldblatt [1987]. Derivability in the reverse direction can be established by more theoretical considerations.  
Goldblatt sketches in the appendix to his book a proof that his own system is complete and decidable, and he demonstrates that his primitive of orthogonality is interdefinable with that of causal connectibility. He derives his result from quantifier elimination for the theory of real closed fields. [Pambuccian 2006] constructs an explicit definition of betweenness and congruence in terms of causal connectibility.\footnote{Beth's definability theorem and a first order strengthening of the  Alexandrov-Zeeman's theorem - according to which every automorphism of Minkowski spacetime preserves congruence relations - already imply that such a definition must exist. [Sklar 1985] says that Malament proved a similar theorem in his PhD thesis; [Malament 2019] attributes a version of the theorem to Robb. Pambuccian [2006] has explicitly found such an adequate definition in terms of lightlike connectibility.} Since the system of Goldblatt [1987] and ours are almost self-evidently sound, we get that a derivation must exist without having to go through the hurdle of providing one. This closes the circle. Our system, that of Goldblatt [1987], and a proper formalization of [Robb 1936] must all be equivalent. 

\begin{remark}
\textit{The system that is most similar to what we are about to propose is the axiomatization of Galilean spacetime sketched by Hartry Field in chapt.4 of [Field 1980]. We use the same methods to form a theory for relativistic spacetimes. The main idea is to employ the already existing systems for affine spaces of dimension four [Tarski and Sczerba, 1979] and for Euclidean geometry [Tarski and Givant, 1999] as basic building blocks of our account. The system that we propose is nominalistic. We will return to the connection with [Field 1980] and the nominalization of physics at the end.}     
\end{remark}

\section{The language}

As in Tarski's system for Euclidean geometry [Tarski and Givant 1999], we assume only one type of  entity in the range of the variables: points. The logical vocabulary consists of the identity symbol `=', negation `$\neg$', conjunction `$\land$' the existential quantifier `$\exists$', and auxiliary symbols. The variables are $x$, $y$, $z$ ... $x'$, $y'$... In defiance of the usual conventions, we use $v_1$, $v_2$, $v_3$ as metavariables ranging over variables to state some schemata. \newline   
The two extralogical primitives are a ternary predicate of betweenness: \newline 

(1)  $Bet(x,y,z)$  \newline 

and a quaternary predicate to compare lengths: \newline   

(2) $<_{\equiv}(x,y,z,w)$ \newline 

that holds of four points $x$, $y$, $z$ and $w$ when the square of the relativistic intervals between x and y is less than that between z and w. By the relativistic interval between two points we mean the geometric quantity that is measured, under appropriated coordinates, by the algebraic expression:  

\begin{center} $\sqrt{(t_1-t_2)^2 - (x_1- x_2)^2 - (y_1- y_2)^2 - (z_1 - z_2)^2}$  \end{center} 

The square of the interval is, therefore, the real valued quantity

\begin{center} $(t_1-t_2)^2 - (x_1- x_2)^2 - (y_1- y_2)^2 - (z_1 - z_2)^2$ \end{center}

We stress that, for reasons of simplicity, we work with the square of the interval rather than the interval. This partitions pairs of points into three categories: those such that the term above is negative, those such that the term above is positive and those such that the term above is zero. This of course embodies a convention about signs. It means that spacelike separated points, for example, will count as having negative `length', since the square of the above quantity is a negative number. As we have mentioned, we do not even attempt to formulate the axioms in primitive notation and, for this reason, the next section is devoted to a battery of definitions.

\begin{remark}
\textit{Our primitive vocabulary contains the predicate $<_{\equiv}(x,y,z,w)$ in lieu of the usual congruence predicate $\equiv(x,y,z,w)$ [Tarski, 1959; Tarski and Givant, 1999]. It is natural to ask whether we could have based our system on congruence instead. The predicate `$<_{\equiv}(x,y,z,w)$' is simply not definable in terms of `$Bet(x,y,z)$' and `$\equiv(x,y,z,w)$' in plane geometry. The Minkwoski two dimensional plane admits of an automorphism of the system of congruence - a bijection that sends congruent segments to congruent segments - but inverts relationships of shorter and longer. Anticipating a bit on our account of representation, we can specifiy it in coordinates as the transformation (x,t) $\mapsto$ (t,x) (swapping of space and time coordinates). In Minkowski spacetime a definition is possible. We can distinguish spacelike segments by the fact that they have congruent orthogonal segments and define `shorter than' as usual. The chain of definitions is cumbersome and we have preferred to adopt `$<_{\equiv}(x,y,z,w)$' as an undefined predicate.                   }     
\end{remark}

\section{A battery of definitions}

Our plan is to describe the geometry of a flat spacetime by specifying axioms that (a) characterize it as a four dimensional vector space and (b) fix the `length' of arbitrary segments. We fix their length by decomposing them into a basis. The length of our initial segment is expressed as a function of those of its projections or components. This requires the machinery of linear algebra. We also need the notion of orthogonality and a development of the theory of proportions; essentially of a device to mimic algebraic computations within the theory. The crucial definition is that of the \textit{orthogonality of two segments}. The development of linear algebra depends on orthogonality rather than orthogonality being defined as in linear algebra. One cannot just start from a given `chosen' basis and define the dot product - or a particular linear form - as a linear function of the components relative to the `preferred' basis. The definitions that follow build up the conceptual tools that we need:

\subsection{Basic definitions}

The first definition introduces the usual congruence predicate `$\equiv$'. 

\begin{enumerate}[start=0,label={(\bfseries D\arabic*):}]
\item $\equiv(x,y,z,w)$ $\leftrightarrow_{df}$ $\neg <_{\equiv} (x,y,z,w) \land \neg <_{\equiv} (z,w,x,y)$  


A lightlike segment is a segment of zero `length': a segment that is congruent to the degenerate segment between a point and itself.  

\item $L(x,y)$ $\leftrightarrow_{df}$ $\equiv (x,y,x,x)$ 

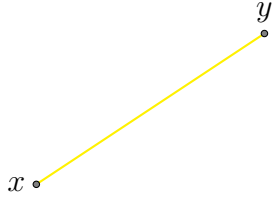
\begin{figure}[H]
    \centering
    \begin{tikzpicture}
    \tkzDefPoint(0,0){x}
    \tkzDefPoint(3,2){v}
    \tkzDrawSegment[color=yellow,thick](x,v)
    \tkzDrawPoints(x,v)
    \tkzLabelPoint[left](x){$x$}
    \tkzLabelPoint[above](v){$y$}
    \end{tikzpicture}
    \caption{The points $x$ and $y$ are lightlike separated (yellow).}
\end{figure}

A spacelike segment is a segment of negative `length' (in blue). 

\item $S(x,y)$ $\leftrightarrow_{df}$ $<_{\equiv}(x,y,x,x)$ \newline \newline

\begin{figure}[H]
    \centering
    \begin{tikzpicture}
    \tkzDefPoint(0,0){x}
    \tkzDefPoint(3,0){w'}
    \tkzDrawSegment[color=blue, thick](x,w')
    \tkzDrawPoints(x,w')
    \tkzLabelPoint[left](x){$x$}
    \tkzLabelPoint[right](w'){$y$}
    \end{tikzpicture}
    \caption{The points $x$ and $y$ are spacelike separated.}
\end{figure}
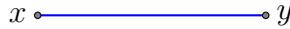

A timelike segment is a segment of positive `length' (in red). 

\item $T(x,y)$ $\leftrightarrow_{df}$ $ <_{\equiv}(x,x,x,y)$

\begin{figure}[H]
    \centering
    \begin{tikzpicture}
    \tkzDefPoint(0,0){x}
    \tkzDefPoint(0,3){v}
    \tkzDrawSegment[color=red,thick](x,v)
    \tkzDrawPoints(x,v)
    \tkzLabelPoint[left](x){$x$}
    \tkzLabelPoint[above](v){$y$}
    \end{tikzpicture}
    \caption{The points $x$ and $y$ are timelike separated.}
\end{figure}
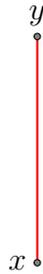

\subsection{Black boxes from the axiomatization of geometry}

The following definitions are imported wholesale from the literature on the axiomatization of geometry and need no further explanation: they define collinearity in terms of betweenness, coplanarity of four points and a preliminary definition of parallelism between the lines on which $xy$ and $zw$ stand. Later we will settle on another definition.       

\item $Coll(x,y,z)$ $\leftrightarrow_{df}$ $Bet (x,y,z) \lor Bet (x,z,y) \lor Bet (y,x,z)$

\item $Copl(x,y,z,w) \leftrightarrow_{df} \exists v  ((Coll(x,y,v) \land Coll(z,v,w)) \lor (Coll(x,z,v) \land Coll(y,v,w)) \lor (Coll(y,z,v) \land Coll(x,v,w))) $

\item $Par_W (x,y,z,w)$ $\leftrightarrow_{df}$ $Copl(x,y,z,w) \land ((Coll(x,y,z) \land Coll(x,y,w)) \lor \neg \exists v (Coll((x,y,v) \land Coll(z,w,y))$ 

\item $Intersect(x,y,z,w)$ $\leftrightarrow_{df}$ $\exists v(Bet(x,v,y) \land Coll(z,w,v))$   

The points $x$, $y$, $z$ and $w$ form a parallelogram when the segments that unite them are pairwise parallel.    

\item $Parallelogram (x,y,z,w)$ $\leftrightarrow_{df}$ $Par_W(x,y,z,w) \land Par_W(x,w,y,z)$

\subsection{Orthogonality}

The main business of this section is to provide a definition of the ternary predicate orthogonality in terms of congruence and betweenness: the segment from x to y is orthogonal to that from x to z. The definition that we give is a definition by cases. The three cases we need to treat separately are (1) the segment from x to y is lightlike, (2) the segment from x to y is spacelike or (3) the segment from x to y is timelike. The strategy is easily grasped by considering how one might define orthogonality in Euclidean geometry. In Euclidean geometry, the orthogonal projection of a point z on a line passing through x and y is simply the \textit{closest point on the line}. This definition can be reproduced wholesale in the case when the segment from x to y is \textit{timelike}. When the segment from x to y is not timelike, the state of affairs is reversed or more complicated. The presence of null and negative lines complicates the business. In all scenarios, a vector from z to some v that falls on the line determined by x and y will give us a right triangle if and only if the segments are orthogonal. The segment from z to v is going to be the hypotenuse of it. Pythagoras has taught us that the square of the hypotenuse is a sum of squares: if the basis of the triangle is spacelike, then the cathetus from x to z is going to contribute negatively to the length of the hypotenuse. This means that the path from z to x is going to be the \textit{longest} straigth path to the line xy.      

\item $Case_1 (x,y,z)$ $\leftrightarrow_{df}$ $S(x,y) \land \forall v (Col (v,y,x) \rightarrow (<_{\equiv} (v,z,x,z) \lor  v=x )) $

\begin{figure}[H]
    \centering
    \begin{tikzpicture}
    \tkzDefPoint(0,0){x}
    \tkzDefPoint(2,0){y}
    \tkzDefPoint(0,2){z}
    \tkzDefPoint(3,0){v}
    
    \tkzDrawSegment[color=blue, thick](x,v)
    \tkzDrawSegment[color=red, thick](x,z)
    \tkzDrawSegment[dashed](v,z)
    \tkzDrawPoints(x, y, z, v)
    \tkzLabelPoint[left](x){$x$}
    \tkzLabelPoint[below](y){$y$}
    \tkzLabelPoint[left](z){$z$}
    \tkzLabelPoint[right](v){$v$}
    \end{tikzpicture}
    \caption{Case 1 of orthogonality}
\end{figure}
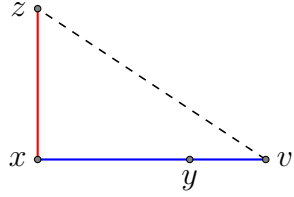

If the basis is timelike, we get the reverse situation. This puts us back, as we noted, in the old Euclidean case. The orthogonal projection of z onto the line xy is the closest point on the line: 

\item $Case_2 (x,y,z)$ $\leftrightarrow_{df}$ $T(x,y) \land \forall v (Col (v,y,x) \rightarrow (<_{\equiv} (x,z,v,z) \lor  v=x))$

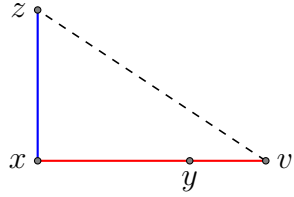
\begin{figure}[H]
    \centering
    \begin{tikzpicture}
    \tkzDefPoint(0,0){x}
    \tkzDefPoint(2,0){y}
    \tkzDefPoint(0,2){z}
    \tkzDefPoint(3,0){v}
    
    \tkzDrawSegment[color=red, thick](x,v)
    \tkzDrawSegment[color=blue, thick](x,z)
    \tkzDrawSegment[dashed](v,z)
    \tkzDrawPoints(x, y, z, v)
    \tkzLabelPoint[left](x){$x$}
    \tkzLabelPoint[below](y){$y$}
    \tkzLabelPoint[left](z){$z$}
    \tkzLabelPoint[right](v){$v$}
    \end{tikzpicture}
    \caption{Case 2 of orthogonality}
\end{figure}

The last case that needs to be treated is when the base of the triangle is lightlike. There are two ways to deal with it. With the two notions of orthogonality at hand, we have enough material to define an \textit{orthogonal basis} and the arithmetic of segments (cf. next section). This apparatus is enough to develop linear algebra. We can then define a nominalistic proxy of the Lorentzian form between two segments. Two orthogonal segments are going to be two segments such that the form gives zero when applied to them. The approach we adopt is more elegant and consists in reducing the third case to the former two. Let us assume again that xy is lightlike and that xz is a candidate to orthogonality. Either (a) z is collinear to x and y or (b) xz is spacelike. We can decompose xy in a spacelike component $xy'$ and a timelike component xw so that xw is orthogonal to xz. At this point, by the distributivity of the Lorentzian product, we see that xz is orthogonal to xy if and only if it is orthogonal to $xy'$. This means that a spacelike xz is orthogonal to a lightlike xy just in case there is a decomposition of xy such that xz is orthogonal to both components in the senses already treated:    

\item $Case_3 (x,y,z)$ $\leftrightarrow_{df}$ $L(x,y) \land [Coll(x,y,z) \lor (\exists w \exists y'(T(w,x) \land S(w,y) \land Case_1(x,w,z) \land Parallelogram(x,w,y,y') \land Case_2(x,z,y'))]$

\begin{figure}[H]
    \centering
    \begin{tikzpicture}
    \tkzDefPoint(0,0){x}
    \tkzDefPoint(2,0){z}
    \tkzDefPoint(2,2.1){z'}
    \tkzDefPoint(2,-1.14){z''}
    \tkzDefPoint(1.25,1.45){y}
    \tkzDefPoint(0,2.1){w}
    \tkzDefPoint(1.25,-0.7){y'}
    \tkzDrawSegment[color=blue, thin](x,z)
    \tkzDrawSegment[color=red, thin](x,w)
    \tkzDrawSegment[thick, color=yellow](x,y)
    \tkzDrawSegment[thin, color=blue](x,y')
    \tkzDrawSegment[dashed](w,y)
    \tkzDrawSegment[dashed](y,y')
    \tkzDrawLine[dashed](w,z')
    \tkzDrawLine[dashed](y',z'')
    \tkzLabelPoint[left](x){$x$}
    \tkzLabelPoint[right](y){$y$}
    \tkzLabelPoint[below](z){$z$}
    \tkzLabelPoint[above](w){$w$}
    \tkzLabelPoint[left](y'){$y'$}
    \tkzDrawPoints(x, y, z, w, y')
    \end{tikzpicture} \newline
    \caption{Case 3 of orthogonality} 
     \end{figure}
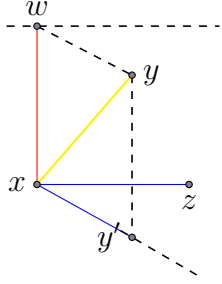 

We can now define orthogonality by a disjunction: 

\item $Orth'(x,y,z)$ $\leftrightarrow_{df}$ $Case_1(x,y,z) \lor Case_2(x,y,z) \lor Case_3(x,y,z)$
\item $Orth (x,y,z)$ $\leftrightarrow_{df}$ $Orth' (x,y,z)  \lor x=y \lor x=z \newline$

and give another definition of parallelism in terms of orthogonality:\footnote{It is useful to stipulate that a point $x$ - that is, a degenerate segment - is respectively orthogonal to lines through $x$ and parallel to lines that do do not pass through $x$. } 

\item $Par(x,y,z,w)$  $\leftrightarrow_{df}$ [$Par_W(x,y,z,w) \land \exists z' (Coll(z, w, z') \land Orth(x,y,z') \land Orth(z', x, z))] \lor x=y \lor z=w $ 

\subsection{Linear Algebra}

We are ready to define the apparatus of linear algebra. When is a segment, or vector, generated from other vectors? To generate a vector $\vec{ox}$ (that stems from a given origin $o$) from other vectors $\vec{o}y$, $\vec{oz}$ and $\vec{ow}$, means that we can reach the `top' $x$ from the `tail' $o$ by travelling along directions that are parallel to the vectors $\vec{oy}$, $\vec{oz}$, $\vec{ow}$.   

\item $Gen_{3D}(o,x,y,z,v)$  $\leftrightarrow_{df}$ $\exists x' \exists y' (Coll(o,x,x') \land Par (x',y',o,y) \land Par (o,z,y',v))$

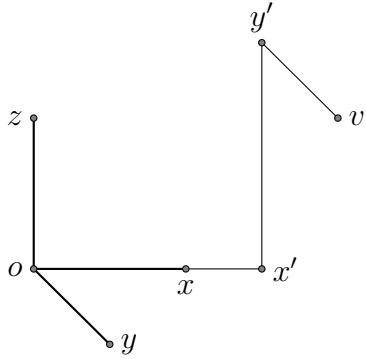
\begin{figure}[H]
    \centering
    \begin{tikzpicture}
    \tkzDefPoint(0,0){o}
    \tkzDefPoint(2,0){x}
    \tkzDefPoint(3,0){x'}
    \tkzDefPoint(3,3){y'}
    \tkzDefPoint(0,2){z}
    \tkzDefPoint(1,-1){y}
    \tkzDefPoint(4,2){v}
    
    \tkzDrawSegment[thick](o,x)
    \tkzDrawSegment[thin](x,x')
    \tkzDrawSegment[thin](x',y')
    \tkzDrawSegment[thick](o,y)
    \tkzDrawSegment[thick](o,z)
    \tkzDrawSegment[thin](y',v)
    \tkzDrawPoints(o,x',x, y, y', z,v)
    \tkzLabelPoint[left](o){$o$}
    \tkzLabelPoint[below](x){$x$}
    \tkzLabelPoint[right](y){$y$}
    \tkzLabelPoint[right](x'){$x'$}
    \tkzLabelPoint[above](y'){$y'$}
    \tkzLabelPoint[left](z){$z$}
    \tkzLabelPoint[right](v){$v$}

    \end{tikzpicture}
    \caption{The point v is a linear combination of $\vec{ox}$, $\vec{oy}$, $\vec{oz}$ in three dimensions.}
\end{figure}

Notation for the generation of more than one vector is easily introduced.

\item $Gen(o,x,y,z,a,b,c,d)$  $\leftrightarrow_{df}$ $Gen_{3D}(o,x,y,z,a) \land Gen_{3D}(o,x,y,z,b) \land  Gen_{3D}(o,x,y,z,c) \land Gen_{3D}(o,x,y,z,d) \newline$  

It is useful to take into account the intermediate steps that are made when moving from $o$ to a given $w$ along one of the specified directions. A specific trajectory may be called a \textit{development} of $w$ from $\vec{ox}$, $\vec{oy}$, $\vec{ox}$ and $\vec{ot}$. There are of course multiple ways to reach w from $o$, depending with which direction one starts from, and also on the different ways of proceeding. We will later impose an axiom that makes parallel trajectories in two separate developments congruent to each other. To express it, we need to refer to these intermediate steps. We read the next predicates as `w can be reached from $x$, $y$, $z$ and $t$ via $x'$, $y'$, $z'$'. 

This is the standard fashion of reaching a point: 

\item $Reached_4(v_0, v_1, v_2, v_3, v_4, x', y', z', w)$ $\leftrightarrow_{df}$ $Coll(v_0, v_1, x')$ $\land$ $Par(v_0, v_2, x', y')$ $\land$ $Par (v_0, v_3, y', z')$ $\land$ $Par(v_0, v_4, w, z')$ $\land$ $Orth(x',v_0, y')$ $\land$ $Orth(y',x',z')$ $\land$ $Orth(z',y',w)$ \newline

\item $Reached(o,x,t,v,w)$ $\leftrightarrow_{df}$ $Coll(o, x,v) \land Orth(v,o,w) \land Par (o,t,v,w) \newline $

An arbitrary development is defined by permuting the order of directions:

\item $Develop(v_0, v_1, v_2, v_3, v_4, x', y', z', v)$ $\leftrightarrow_{df}$ $\bigvee Reached_4(v_0, v_{\sigma(1)}, v_{\sigma(2)}, v_{\sigma(3)}, v_{\sigma(4)}, x', y', z', v) \newline $
[where $\sigma$ ranges over permutations of the set of indices $\{1,2,3,4\}]$. \\

Generation in four dimension is defined in terms of development: 

\item $Gen_{4D}(o,x,y,z,t,v)$  $\leftrightarrow_{df}$ $\exists x' \exists y' \exists z' Develop(o,x,y,z,t,x',y',z',v)$ \\

A basis is for us a quintuple of points. It consists of an origin o and four points that determine four mutually orthogonal directions in spacetime. 

\item $Basis (o, x, y, z, t)$ $\leftrightarrow_{df}$ $T(o,t) \land S(o,x) \land S(o,y) \land S(o,z) \land Orth(o,t,x) \land Orth(o,t,y) \land Orth(o,t,z) \land Orth(o,x,y) \land Orth(o,x,z) \land Orth(o,y,z) $

\subsection{Segments of opposite length}

The next definition plays a crucial role in the economy of our system. It is a key ingredient of all our main axioms. It defines the relation that obtains between a timelike and spacelike vector when the square of the interval, or the `length' of these segments, differ only in term of `sign': when they are of equal absolute value. To define the notion we have to transport one of the segments to a congruent one orthogonal to the other. At this point we can call them of opposite length if their sum is a null vector: their contributions to the hypotenuse cancel out.     

\item $Opp(x,y,z,w)$ $\leftrightarrow_{df}$ $\exists w' \exists v (Orth(x,y,w') \land \equiv (x, w', z, w) \land Par(w', v, x,y ) \land \equiv (w', v, x, y) \land L(x, v))$

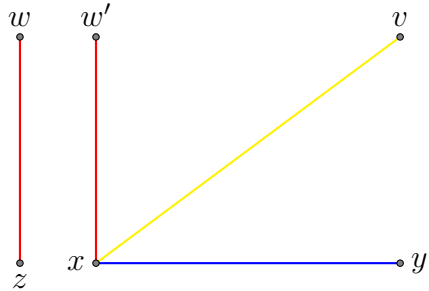
\begin{figure}[H]
    \centering
    \begin{tikzpicture}
    \tkzDefPoint(0,0){x}
    \tkzDefPoint(4,0){w'}
    \tkzDefPoint(0,3){y}
    \tkzDefPoint(4,3){v}
    \tkzDefPoint(-1,0){z}
    \tkzDefPoint(-1,3){w}
    \tkzDrawSegment[color=red,thick](x,y)
    \tkzDrawSegment[color=blue, thick](x,w')
    \tkzDrawSegment[color=yellow,thick](x,v)
    \tkzDrawSegment[color=red,thick](z,w)
    \tkzDrawPoints(x,w',y, z, w, v)
    \tkzLabelPoint[left](x){$x$}
    \tkzLabelPoint[right](w'){$y$}
    \tkzLabelPoint[above](y){$w'$}
    \tkzLabelPoint[above](v){$v$}
    \tkzLabelPoint[below](z){$z$}
    \tkzLabelPoint[above](w){$w$}
    \end{tikzpicture}
    \caption{The points $x,y,z$ and $w$ are opposite.}
\end{figure}

\end{enumerate}

\subsection{Streckenrechnung}

Following the ideas of Hilbert [1899], one can define algebraic operations on the points of a line. Given a line in Euclidean space, fix two arbitrary points to play the role of the null element 0 and the neutral element 1. We can use the method to define addition and multiplication of line segments that are collinear to 0 and 1 in such a way that the line satisfies the axioms of a real closed field. Of course lines in different models will lead to different fields or rings. Each of them, however, simulates well enough the familiar field of real numbers (an informal presentation of the construction is in [Hartshorne 2000, ch.4]). Segment arithmetic is of crucial importance for Field's Program. It allows us to translate numerical statements about real numbers into purely geometrical ones. It also plays an important role in our attempt to construct coordinate systems within the geometrical theory. Tarski and Sczerba [1979] use it construct and classify the coordinatizations of various spaces modulo the `passive transformations' between them; they fix an origin in an affine space $A$ with certain properties, find a line $\ell$ living inside it that satisfies the field axioms, they define vector operations among the points and make the structure $<A,F>$ into a vector field $V$ over a field $F'$.  

The arithmetic of segments is now needed to compute within the theory the `length' of the hypotenuse of a right triangle relative to the `length' of the sides. It will allow us to postulate the existence of a third segment whose `length' is the sum or the product of the length of any two given segments. It will also allow us to postulate a segment on any given line such that its `length' is the square or the square root of the length of any given segment. The apparatus is imported as a block from [Schwabhäuser, Szmielew and Tarski 1983]. But their definitions are meant in the context of Euclidean geometry. We will therefore need to restrict the variables so that all the segments involved are spacelike. Some of their initial definitions can be restated more simply for our purposes in terms of congruence:                 

\begin{enumerate}[start=20,label={(\bfseries D\arabic*):}]

\item $Add(x,y,z, w, v, l)$ $\leftrightarrow_{df}$ $ S(x,y) \land S(z,w) \land \exists v_1 \exists v_2 \exists v_3 (Bet(v_1,v_2,v_3) \land \newline \equiv(x,y,v_1,v_2) \land \equiv(z,w,v_2,v_3) \land \equiv(v,l,v_1,v_3))$

\item $Dupl(x,y,z,w)$ $\leftrightarrow_{df}$ $Add(x, y, x, y, v, l)$ 

\item $Square(x,y,z,w)$ $\leftrightarrow_{df}$ $\exists v_1 \exists v_2 \exists v_3 (\lnot Bet(v_1, v_2, v_3) \land Orth(v_1, v_2, v_3) \land \newline \equiv(x,y,v_1,v_2) \land \equiv(x, y,v_1,v_3) \land Dupl(z,w,v_2, v_3))$ 

\item $Sqrt(x,y,z,w)$ $\leftrightarrow_{df}$ $Square(z,w,x,y) $ 
\end{enumerate}

The next definition formalizes of that of Hartshorne [2000, p. 170]\footnote{An alternative approch, purely in terms of betweenness, can be found in the treatise of [Schwabhäuser, Szmielew and Tarski 1983, p. 160]}: 

\begin{enumerate}[start=24,label={(\bfseries D\arabic*):}]
\item $Product(o, e, x, y, z, w, l,v)$ $\leftrightarrow_{df}$ $\exists y' \exists w' \exists v'$ ($Orth(e,o,y')$ $\land$ $Orth(w',o, v')$ $\land$ $\equiv (o,w',z,w)$ $\land$ $\equiv(e,y',x,y)$ $\land$  $\equiv(l,v,w',v')$
 
The next definition has no particular intrinsic significance and merely abbreviates the result of the calculation in axiom \textbf{(SUM1)}. It is included only to avoid cluttering the axiom and to improve readability. 

\item $Remterm (o,e, x,y,w, v_5)$ $\leftrightarrow_{df}$  $\exists v_1 \exists v_2 \exists v_3 \exists v_4 ((Square(w,y,x,v_1) \land \\ Prod(o,e,x,w, w,y, x, v_2) \land Dup(x, v_2, x, v_3) \land Add(x, v_1, x, v_3, x, v_4) \land Sqrt( x, v_4, x, v_5)) $   

To extend the calculus of segments to timelike vectors the simplest approach is to move back and forth using opposites. For instance, the sum and product of two points on a timelike line is the opposite of the sum and product of two opposite segments on a spacelike line.

\item $Add'(x,y,z, w, v, l)$ $\leftrightarrow_{df}$ $\exists x'\exists y' \exists z'\exists w' \exists v' \exists l'$ $(Opp(x',y',x,y)$ $\land$ $Opp(z',w',z,w)$ $\land$ $Opp(v',l', v, w)$ $\land$ $Add(x',y',z',w',v',l'))$ 

\item $Prod'(o,e, x, y, z,w, v, l)$ $\leftrightarrow_{df}$  $\exists x' \exists y' \exists z'\exists w' \exists v' \exists l'(Opp(x',y',x,y) \land Opp(z',w',z,w) \land Opp(v',l',v,l) \land Prod(o,e, x,y,z,w,v,l))$
 
\item $Remterm' (o,e, x,y,w, v_5)$ $\leftrightarrow_{df}$ $\exists x'\exists y' \exists w' \exists v_5'(Opp(x',y',x,y)$ $\land$ $Opp(x',v_5',x,v_5) \land Remterm(o,e, x,y,w,v_5))$
\end{enumerate}

To define a product operation on a timelike and a spacelike segment we proceed in a similar fashion using the notion of opposites.  

\begin{enumerate}[start=29,label={(\bfseries D\arabic*):}]
\item $Prod''(o,e, x, y, z, w, v, l)$ $\leftrightarrow_{df}$ $T(x,y)$ $\land$ $S(z,w)$ $\land$ $\exists x'\exists y' \exists v' \exists l' (Opp(x',y',x,y)$ $\land$ $Prod(o,e, x',y',z, w,v',l') \land Opp(v,l,v',l'))$
 
\end{enumerate}

\section{An overview of the axioms}

The system can be divided into six groups of axioms. The first axioms govern the notion of betweenness on a line. This part consists of the axioms for a four dimensional affine space, as formalised in [Tarski and Sczerba, 1979] or in [Schwabhäuser, Szmielew and Tarski, 1983, p. 415-416]. We omit figures for them. We call the  second part  dimensionality axioms: they assert the existence of a basis for every choice of an origin; that every point can be reached or `generated' through alternative paths; and finally, that o, x, y, and z form a basis for a Euclidean subspace. Three segments in this basis are spacelike; a fourth is timelike. The set of points that is spanned by the orthogonal spacelike ones must always form a three dimensional Euclidean space. This requirement is ensured by postulating that these points obey the axioms of Euclidean geometry in [Tarski and Givant, 1999].\footnote{An alternative approach would be an axiom that says that the set of points which have a fixed positive distance to an origin satisfy the axioms of hyperbolic geometry, i.e. the axioms of Euclidean geometry where Euclid axiom is replaced by its negation (for details of the construction see the last chapter of [Malament, unpublished]).} A third group of axioms constrains the length of arbitrary segments in terms of the lengths of the components. The fourth group consists of construction axioms: they postulate the existence of segments on a given line that match any other line - either in the sense that they are congruent or opposites. We then have a fifth group of axioms concerning formal properties of the relations employed. We conclude with the axiom schema of continuity and an axiom for density.  

\begin{definition}
$\ulcorner Tarski\urcorner$ abbreviates the conjunction of Tarski's axioms for three dimensional Euclidean geometry in [Tarski and Givant, 1999] with the exception of (a) the axiom schema of continuity, (b) the axioms of affine geometry, (c) the Five-Segment Axiom [Ax. 5] and (d) [Ax.23] and [Ax.24]  .  
\end{definition} 

\begin{definition} 
We define the restriction $\ulcorner \phi^{E(o,x,y,z)} \urcorner$ of $\phi$ to the space generated by o, x, y and z by induction on the complexity/construction of formulae:  \newline

(1) if $\phi$ is atomic, then $\phi^{E(o,x,y,z)}$ is $\phi$. \newline

(2) $\ulcorner\neg \phi\urcorner^{E(o,x,y,z)}$ is $\neg \ulcorner\phi^{E(o,x,y,z)}\urcorner$ \newline

(3) $\ulcorner \phi \land \psi\urcorner ^{E(o,x,y,z)}$ is $\ulcorner \phi^{E(o,x,y,z)}\urcorner$ $\land$ $\ulcorner \psi^{E(o,x,y,z)}\urcorner$ \newline

(4) $\ulcorner \exists v \,  \phi\urcorner^{E(o,x,y,z)}$ is $\exists v \, (Gen_{3D}(o,x,y,z,v) \land \ulcorner\phi^{E(o,x,y,z)}\urcorner)$ \newline

\end{definition} 

We write $\vdash \phi$ to mean that the universal closure of $\phi$ is an axiom. 

\section{The axiomatic system}

\subsection*{Axioms for affine space }
\begin{enumerate}[start=0,label={(\bfseries AFF\arabic*):}]

\item $\vdash Bet(x,y,x) \rightarrow x = y$
\item $\vdash Bet(x,y,z) \land Bet(y,z,u) \land y \neq z \rightarrow Bet(x,y,u)$
\item $\vdash Bet(x,y,z) \land Bet(x,y,u) \land x \neq y \rightarrow Bet(y,z,u) \lor Bet(y,u,z)$
\item $\vdash \exists x(Bet(x,y,z) \land x \neq y)$
\item $\vdash Bet(x,t,u) \land Bet (y,u,z) \rightarrow (Bet(x,v,y) \land Bet(z,t,v))$
\item $\vdash Bet(x,u,t) \land Bet(y,u,z) \land x \neq u \rightarrow \exists  v \exists w (Bet(x,y,v) \land Bet(x,z,w) \land Bet(v,t,w))$
\end{enumerate}

\subsection*{Dimension axioms}

The following axioms ensure that the entire space is a four dimensional vector space. The first axiom says that, for any choice of a point $o$ as the origin, there are other four points such that they form a basis of the space. 
\begin{enumerate}[start=0,label={(\bfseries A\arabic*):}]


\item $\vdash \forall o \, \exists x \, \exists y \, \exists z \, \exists t \; Basis (o,x,y,z,t)$ 

Any orthogonal segments $\vec{ox}$ and $\vec{ot}$ can be supplemented to a basis: 

\item  $\vdash T(o,t) \land Orth(o,x,t) \rightarrow \exists y \exists z \, Basis(o,x,y,z,t) $ 

The next axiom asserts that every basis generates every point.

\item $\vdash Basis(o,x,y,z,t) \rightarrow Gen_{4D} (o,x,y,z,t,v)$ \newline 


The space spanned by the spatial subbasis obeys the axioms of Tarski. It is a three dimensional Euclidean space. The schema of continuity and the axioms for betweenness are assumed later for all lines.

\item $\vdash Basis (o, x, y, z, t) \rightarrow Tarski^{E(o,x,y,z)}$ \newline 

Typical axioms of Eucliden geometry postulate the congruence of certain triangles under hypotheses about the congruence of certain angles and certain sides. In the system described in [Tarski and Givant, 1999] these criteria of congruence are derived from a single Five-Segment Axiom [Ax. 5]. It is convenient to adapt it to our system by assuming that the two triangles to be compared can come from different spacelike hyperplanes. The abbreviation below is self-explanatory.  

\item $\vdash (Basis (o, x, y, z, t) \, \land \, Basis (o', x', y', z', t') \, \land \, Gen(o, x, y, z, t,a,b,c,d) \, \land \, Gen(o', x', y', z', t',a',b',c',d')) \, \rightarrow \,$ Five-Segment Axiom$(a,b,c,d,a',b',c',d') $ \newline

This axiom asserts that alternative paths to the same point consist of congruent segments. This implies that the lengths of the components of a segment depend only on the basis and not on the development.    

\item $\vdash Reached(o,x,y,v,w) \rightarrow \exists v' (Reached (o,y,x,v',w) \, \land \equiv(o,v',v,w) \land \newline \equiv (o,v,v',w)) $

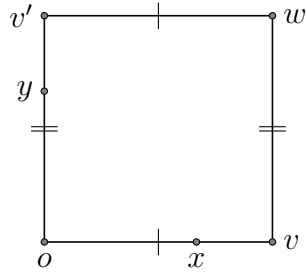
\begin{figure}[H]
    \centering
    \begin{tikzpicture}
    \tkzDefPoint(0,0){o}
    \tkzDefPoint(2,0){x}
    \tkzDefPoint(0,2){y}
    \tkzDefPoint(3,0){v}
    \tkzDefPoint(0,3){v'}
    \tkzDefPoint(3,3){w}
    \tkzDrawSegment(o,v')
     \tkzDrawSegment(w,v')
      \tkzDrawSegment(o,v)
      \tkzDrawSegment(v,w)
    \tkzDrawPoints(o,x,y,v,w,v')
    \tkzMarkSegments[mark= ||,size=5pt](o,v' v,w)
    \tkzMarkSegments[mark= |,size=5pt](o,v v',w)
    \tkzLabelPoint[below](o){$o$}
    \tkzLabelPoint[below](x){$x$}
    \tkzLabelPoint[left](y){$y$}
    \tkzLabelPoint[right](v){$v$}
    \tkzLabelPoint[left](v'){$v'$}
    \tkzLabelPoint[right](w){$w$}
    \end{tikzpicture}
    \caption{Axiom (A5)}
\end{figure}

Linear algebra requires that the sum of two vectors be unique. The next axiom imposes that a vector have a unique decomposition.\footnote{The analogy between $Reached(o,x,t,v,w)$ and the operation of vector sum is imperfect because it does not distinguish bewteen $\vec{a} + \vec{b}$ and $\vec{a} - \vec{b}$.}   

\item $\vdash (Reached(o,x,t,v,w) \land Reached(o,x,t,v',w)$ $\land$ \newline  \centering $ (Coll(o,x,w) \lor \lnot Intersect(v, v', o, x)) \rightarrow v=v'$

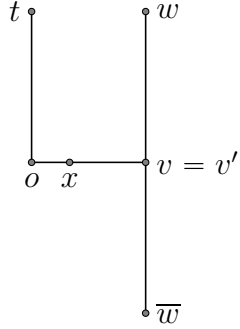
\begin{figure}[H]
\centering
\begin{tikzpicture}

    \tkzDefPoint(0,0){o}
    \tkzDefPoint(0.5,0){x}
   \tkzDefPoint(1.5,0){v}
   \tkzDefPoint(0,2){t}
   \tkzDefPoint(1.5,2){w}
   \tkzDefPoint(1.5,-2){w'}
    \tkzDrawSegment(o,t)
    \tkzDrawSegment(o,v)
    \tkzDrawSegment(v,w)
    \tkzDrawSegment(v,w')
    \tkzDrawPoints(o,x,v,t,w,w')
    \tkzLabelPoint[below](o){$o$}
    \tkzLabelPoint[below](x){$x$}
    \tkzLabelPoint[left](t){$t$}
    \tkzLabelPoint[right](v){$v=v'$}
    \tkzLabelPoint[right](w){$w$}
    \tkzLabelPoint[right](w'){$\overline{w}$}
    \end{tikzpicture}
\caption{Axiom (A6)}
\end{figure}

To extend (A5) and (A6) to the uniquess of sums of more than two vectors, that is of developments in three or four steps, we need (A7): 

\item $\vdash (Orth(o,x,z) \land Orth(o,y,z) \land \exists v Reached (o,x,y,v,w)) \rightarrow Orth(o,z,w)$

\begin{figure}[H]
\centering
\begin{tikzpicture}
\tkzDefPoint(0,0){o}
\tkzDefPoint(2,0){x}
\tkzDefPoint(1,0){v}
\tkzDefPoint(-0.7,-0.7){y}
    \tkzDefPoint(0,2){z}
    \tkzDefPoint(0.4,-0.7){w}
    \tkzDrawSegment[thin](o,x)
    \tkzDrawSegment[thin](o,z)
    \tkzDrawSegment[thin](o,y)
    \tkzDrawSegment[thin](o,w)
    \tkzDrawSegment[dashed](v,w)
    \tkzLabelPoint[left](o){$o$}
    \tkzLabelPoint[left](y){$y$}
    \tkzLabelPoint[right](x){$x$}
    \tkzLabelPoint[left](z){$z$}
    \tkzLabelPoint[right](w){$w$}
    \tkzLabelPoint[above](v){$v$}
    \tkzDrawPoints(o,x,y,z,v,w)
    \end{tikzpicture} 
    \caption{Axiom (A7)} 
    \end{figure}
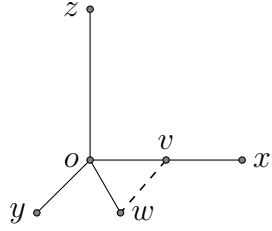 
\end{enumerate}

Minkowski spacetime cannot be accurately described unless we relate the ordering on a generic line with our foliations into a timelike line and a spacelike hyperplane. Betwenness on a lines corresponds to another basic notion of linear algebra: scalar multiplication. Two points are on the same line if and only if the components of one are scaled with respect to the components of the other by the same factor $\lambda$. Axiom (A8) reads:

\begin{enumerate}[start=8,label={(\bfseries A\arabic*):}]
\item $\vdash Reached(o,x,t,x,r)$ $\land$ $Reached(o,t, x, t, r)$ $\land$ $Reached(o,x,t,x',r')$ $\land$ $Reached(o,t,x,t',r')$ $\land$ $T(o,t)$ $\land$ $o\neq{e}$ $\rightarrow [Bet(o,r,r')$ $\leftrightarrow (Bet(o,x,x')$ $\land$ $Bet(o,t,t')$  $\land$ $\exists z$ $\exists z' (S(z,z')$ $\land$ $Product(o,e, x,z,z',o,x')$ $\land Product''(o,e, t,z,z',o, t')))$]

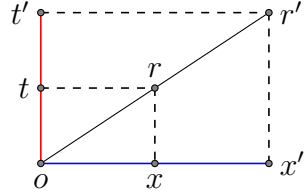
\begin{figure}[H]
    \centering
    \begin{tikzpicture}
    \tkzDefPoint(0,0){o}
    \tkzDefPoint(3,0){x'}
    \tkzDefPoint(1.5,0){x}
    \tkzDefPoint(0,2){t'}
    \tkzDefPoint(0,1){t}
    \tkzDefPoint(3,2){r'}
    \tkzDefPoint(1.5,1){r}
    \tkzDrawSegment[dashed](r',x')
    \tkzDrawSegment[dashed](r,x)
    \tkzDrawSegment[dashed](t,r)
    \tkzDrawSegment[dashed](t',r')
    \tkzDrawSegment[color=blue](o,x')
    \tkzDrawSegment[thin](o,r')
    \tkzDrawSegment[color=red](o,t')
    \tkzDrawPoints(o,x,x',t,t',r,r')
    \tkzLabelPoint[below](o){$o$}
    \tkzLabelPoint[below](x){$x$}
    \tkzLabelPoint[right](x'){$x'$}
    \tkzLabelPoint[left](t'){$t'$}
    \tkzLabelPoint[left](t){$t$}
    \tkzLabelPoint[above](r){$r$}
    \tkzLabelPoint[right](r'){$r'$}
    \end{tikzpicture}
    \caption{Axiom (A8)}
\end{figure}

The last axiom of the present section postulates that to two orthogonal vectors can indeed always be associated a sum. 

\item $\vdash Orth(o,x,t) \rightarrow \exists w$ $Reached(o,x,t,x,w) $ 

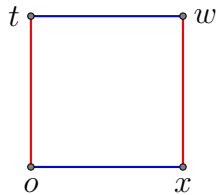
\begin{figure}[H]
    \centering
    \begin{tikzpicture}
    \tkzDefPoint(0,0){o}
    \tkzDefPoint(2,0){x}
    \tkzDefPoint(0,2){t}
    \tkzDefPoint(2,2){w}
    \tkzDrawSegment[color=red,thick](o,t)
    \tkzDrawSegment[color=blue, thick](o,x)
    \tkzDrawSegment[color=blue, thick](t,w)
    \tkzDrawSegment[color=red, thick](x,w)
    \tkzDrawPoints(o,x,t,w)
    \tkzLabelPoint[below](o){$o$}
    \tkzLabelPoint[below](x){$x$}
    \tkzLabelPoint[left](t){$t$}
    \tkzLabelPoint[right](w){$w$}
    \end{tikzpicture}
    \caption{Axiom (A9)}
\end{figure}
\end{enumerate}

\begin{remark}
\textit{The relative ugliness of the axioms in this section can be remedied somewhat by introducing the notation of linear algebra. This may improve their readability as well. For example, axioms (A7) and (A8) assert the existence and uniqueness of the sum $\vec{ox} + \vec{oy}$ of the  orthogonal vectors $\vec{ox}$ and $\vec{oy}$. Axiom (A4) asserts the familiar axiom of a vector space: $\vec{ox} + \vec{oy} =  \vec{oy} + \vec{ox}$ (commutativity of addition). The axiom (A7) is a basic consequence of the distributivity of the Lorentzian product: the statement$: \ulcorner (\vec{oz} \bullet \vec{ox} = 0)$ $\land$ $(\vec{oz} \bullet \vec{oy} = 0) \rightarrow (\vec{oz} \bullet (\vec{ox} + \vec{oy}) = 0) \urcorner$. (A8) concerns scalar multiplication. }     
\end{remark}

\subsection*{Summation axioms}

We have postulated axioms that assert the existence of bases and permit a decomposition of arbitrary segments into orthogonal components. We now need axioms for the metrical structure. We want a segment extending a spacelike segment to be spacelike and shorter and a segment extending a timelike segment to be timelike and longer. We want, moreover, to be able to compute the length of a segment from that of its components. \newline

Two segments that are both opposite to a third are congruent.

\begin{enumerate}[start=0,label={(\bfseries SUM\arabic*):}]

\item $ \vdash Opp(x,t,z,w) \land Opp (x,t,z',w') \rightarrow \; \; \equiv(z,w,z',w')$

\begin{figure}[H]
    \centering
    \begin{tikzpicture}
    \tkzDefPoint(0,0){x}
    \tkzDefPoint(0,3){y}
    \tkzDefPoint(1,0){z}
    \tkzDefPoint(3,0){w}
    \tkzDefPoint(-1,0){z'}
    \tkzDefPoint(-3,0){w'}
    \tkzDrawSegment[color=red](x,y)
    \tkzDrawSegment[color=blue](z,w)
    \tkzDrawSegment[color=blue](z',w')
    \tkzDrawPoints(x,y,z,w,z',w')
    \tkzLabelPoint[below](x){$x$}
    \tkzLabelPoint[above](y){$t$}
    \tkzLabelPoint[below](w){$w$}
    \tkzLabelPoint[below](z){$z$}
    \tkzLabelPoint[below](z'){$z'$}
    \tkzLabelPoint[below](w'){$w'$}
    \end{tikzpicture}
    \caption{Axiom (SUM0)}
\end{figure}
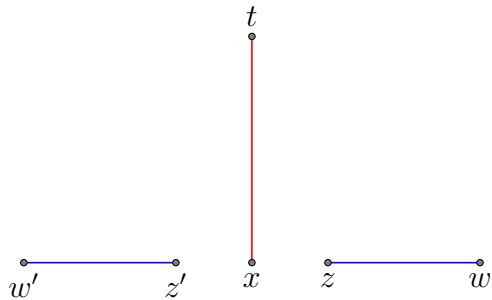

The following two axioms employ the arithmetic of segments that we have defined in section \textbf{3.6} to calculate the length of a segment from its decomposition onto a given basis. Every vector $\vec{xv}$ can be construed as the sum of a spacelike component and a timelike component. We treat separately the case in which (1) the spacelike segment is longer in absolute value \textbf{(SUM1)} and that in which (2)  the timelike segment is longer in absolute value \textbf{(SUM2)}. The basis $\vec{xy}$ is longer in absolute value than the timelike component $\vec{xz}$ if and only if there is a point between $x$ and $y$ that is of opposite length to $\vec{xz}$. If the spacelike side is longer in absolute value, then the hypothenuse $\vec{yz}$ of the right triangle xyz is spacelike. If the timelike side is longer in absolute value, then the hypothenuse $\vec{yz}$ is timelike. To quantify more precisely the length of the hypothenuse $\vec{yz}$ in all cases we need a calculation. Suppose the spacelike segment $\overline{xy}$ is orthogonal to a timelike segment $\overline{xz}$ (see figure 11). Call the length of $\overline{xz}$ $A$ and the length of $\overline{xy}$ $B$. Suppose a segment $\overline{xw}$ is opposite to $\overline{xz}$, whose length we call $D$. Call $E$ the length of $\overline{wy}$. Then, the length $C$ of the resultant vector $\overline{xv}$ is conguent to the hypotenuse $\overline{xz}$. By Pythagoras's theorem, the length of the hypothenuse is:\footnote{The  proportion between $\overline{xv}$ and $\overline{wy}$ that results from the calculation is exactly what is expressed by the predicate $Remterm(x,y,w,v_5)$ introduced without explanation in \textbf{(D28)}} \newline 

(*) $C^2 = A^2 + B^2 = A^2 + (D+E)^2 = \cancel{A^2} +\cancel{D^2} + E^2 + 2DE =  E^2 + 2DE$.\newline

\item $\vdash (S(o, v)$ $\land$ $S(o,x)$  $\land$  $T(o,t)$ $\land$ $Reached(o,x,t,x,v)$ $\land$ $o\neq{e}$ $\rightarrow \;  \exists w  \exists v_5 (Bet(o,w,x) \land Opp(o,w,o,t) \land Remterm(o,e, o,x,w,v_5) \land \equiv (o, v, o, v_5) )$

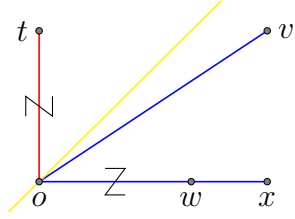
\begin{figure}[H]
    \centering
    \begin{tikzpicture}
    \tkzDefPoint(0,0){x}
    \tkzDefPoint(3,0){y}
    \tkzDefPoint(0,2){z}
    \tkzDefPoint(2,0){w}
    \tkzDefPoint(3,2){v}
    \tkzDefPoint(2,2){l}
    \tkzDrawSegment[color=blue](x,w)
    \tkzDrawSegment[color=blue](w,y)
    \tkzDrawSegment[color=red](x,z)
    \tkzDrawSegment[color=blue](x,v)
    \tkzDrawLine[color=yellow](x,l)
    \tkzDrawPoints(x,y,z,w,v)
    \tkzMarkSegments[mark= z,size=5pt](x,z x,w)
    \tkzLabelPoint[below](x){$o$}
    \tkzLabelPoint[below](y){$x$}
    \tkzLabelPoint[left](z){$t$}
    \tkzLabelPoint[right](v){$v$}
    \tkzLabelPoint[below](w){$w$}
    \end{tikzpicture}
    \caption{Axiom (SUM1)}
\end{figure}

A similar calculation can be made when $\vec{xv}$ is timelike.    

\item $\vdash (T(o,v) \land T(o,t) \land S(o,x) \land Reached(o,x,t,x,v)$ $\land$ $o\neq{e}$) $\rightarrow \;$  $\exists w \newline \exists v_5 (Bet(o,w,t)\land Opp(o,w,o,x) \land Remterm'(o,e,o,t,w,v_5) \land \equiv (o, v, o, v_5))$

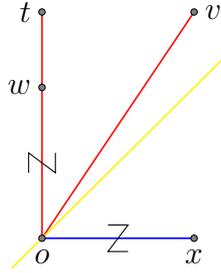
\begin{figure}[H]
    \centering
    \begin{tikzpicture}
    \tkzDefPoint(0,0){x}
    \tkzDefPoint(0,3){y}
    \tkzDefPoint(2,0){z}
    \tkzDefPoint(0,2){w}
    \tkzDefPoint(2,3){v}
    \tkzDefPoint(2,2){l}
    \tkzDrawSegment[color=red](x,w)
    \tkzDrawSegment[color=red](w,y)
    \tkzDrawSegment[color=blue](x,z)
    \tkzDrawSegment[color=red](x,v)
    \tkzDrawLine[color=yellow](x,l)
    \tkzDrawPoints(x,y,z,w,v)
    \tkzMarkSegments[mark= z,size=5pt](x,z x,w)
    \tkzLabelPoint[below](x){$o$}
    \tkzLabelPoint[left](y){$t$}
    \tkzLabelPoint[below](z){$x$}
    \tkzLabelPoint[right](v){$v$}
    \tkzLabelPoint[left](w){$w$}
    \end{tikzpicture}
    \caption{Axiom (SUM2)}
\end{figure}

The sum of segments of opposite length gives lightlike vectors:

\item $\vdash Opp(o,x,o,t) \land Orth(o,x,t) \land Par(o, t, x, v)) \rightarrow \;  L(o,v) $

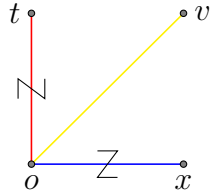
\begin{figure}[H]
    \centering
    \begin{tikzpicture}
    \tkzDefPoint(0,0){x}
    \tkzDefPoint(2,0){y}
    \tkzDefPoint(0,2){z}
    \tkzDefPoint(2,2){v}
    \tkzDrawSegment[color=blue](x,y)
    \tkzDrawSegment[color=red](x,z)
    \tkzDrawSegment[color=yellow](x,v)
    \tkzDrawPoints(x,y,z,v)
    \tkzMarkSegments[mark= z,size=5pt](x,z x,y)
    \tkzLabelPoint[below](x){$o$}
    \tkzLabelPoint[below](y){$x$}
    \tkzLabelPoint[left](z){$t$}
    \tkzLabelPoint[right](v){$v$}
    \end{tikzpicture}
    \caption{Axiom (SUM3)}
\end{figure}

The following two axioms assure us that continuing on a spacelike line we traverse progressively shorter segments, as we move towards infinity.    

\item $\vdash Bet(x,y,z) \land S(x,z) \rightarrow \; \;  <_{\equiv} (x, z, x, y)$

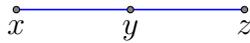
\begin{figure}[H]
    \centering
    \begin{tikzpicture}
    \tkzDefPoint(0,0){x}
    \tkzDefPoint(1.5,0){y}
    \tkzDefPoint(3,0){z}
    \tkzDrawSegment[color=blue](x,z)
    \tkzDrawPoints(x,y,z)
    \tkzLabelPoint[below](x){$x$}
    \tkzLabelPoint[below](y){$y$}
    \tkzLabelPoint[below](z){$z$}

    \end{tikzpicture}
    \caption{xz is shorter than xy (identical figure for (SUM5)}
\end{figure}

\item $\vdash Bet(x,y,z) \land S(x,y) \rightarrow \; \; <_{\equiv} (x, z, x, y)$ \\

We can obtain a similar result for timelike segments. Continuing on a timelike line, we traverse longer and longer segments. We can derive this result from principles relating opposites. Let us remind ourselves that orthogonal opposites cancel i.e., they give a lightlike segment when summed. We postulate (SUM6) that the opposite of a longer timelike segment must be shorter - more in the negative - and \textit{vice versa}.   

\item $\vdash <_{\equiv} (x,x',z,w) \, \land \, Opp(x,x',t,t') \, \land \, Opp(z,w,z',w') \rightarrow \; \; <_{\equiv} (z', w', t, t')$

\begin{figure}[H]
    \centering
    \begin{tikzpicture}
    \tkzDefPoint(0,0){x}
    \tkzDefPoint(4,0){y}
    \tkzDefPoint(-2,0){x'}
    \tkzDefPoint(-2,4){y'}
    \tkzDefPoint(0,2){z}
    \tkzDefPoint(2,2){w}
    \tkzDefPoint(-1,2){z'}
    \tkzDefPoint(-1,4){w'}
    \tkzDrawSegment[color=blue](x,y)
    \tkzDrawSegment[color=red](x',y')
    \tkzDrawSegment[color=blue](z,w)
    \tkzDrawSegment[color=red](z',w')
    \tkzDrawPoints(x,y,z,w,x',y',z',w')
    \tkzLabelPoint[below](x){$x$}
    \tkzLabelPoint[below](y){$x'$}
    \tkzLabelPoint[below](w){$w$}
    \tkzLabelPoint[below](z){$z$}
    \tkzLabelPoint[right](x'){$t$}
    \tkzLabelPoint[above](y'){$t'$}
    \tkzLabelPoint[below](z'){$z'$}
    \tkzLabelPoint[above](w'){$w'$}
    \end{tikzpicture}
    \caption{If $\overline{xy}$ is shorter than $\overline{zw}$, then the opposite of $\overline{zw}$ is shorter than the opposite of $\overline{xy}$.}
\end{figure}
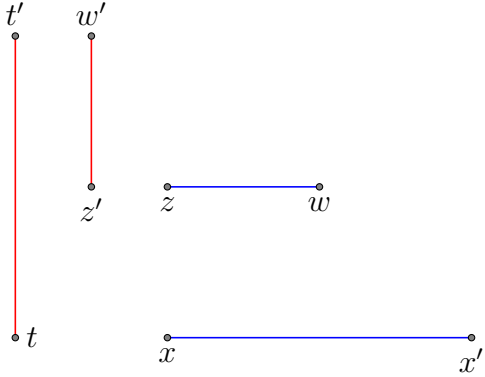

(SUM 6) tells us little about the arrangement of opposite segments on a spacelike and a timelike line. Using our primitive of betweenness, we need to postulate (SUM 7) and (SUM 8) that the ordering of the opposites on a segment mirrors that of the original segment:  

\item $\vdash (Opp(x,y,x',y') \land Bet (x,y,z)) \rightarrow \exists z' (Bet(x',y',z') \land Opp (x',z',x,z))$

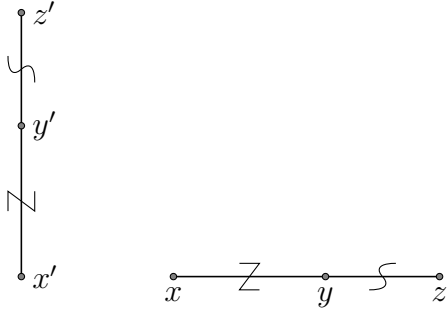
\begin{figure}[H]
    \centering
    \begin{tikzpicture}
    \tkzDefPoint(0,0){x}
    \tkzDefPoint(2,0){y}
    \tkzDefPoint(3.5,0){z}
    \tkzDefPoint(-2,0){x'}
    \tkzDefPoint(-2,3.5){z'}
    \tkzDefPoint(-2,2){y'}
    \tkzDrawSegment(x,z)
    \tkzDrawSegment(x',z')
    \tkzDrawPoints(x,y,z,x',y',z')
    \tkzLabelPoint[below](x){$x$}
    \tkzLabelPoint[below](y){$y$}
    \tkzLabelPoint[below](z){$z$}
    \tkzLabelPoint[right](x'){$x'$}
    \tkzLabelPoint[right](y'){$y'$}
    \tkzLabelPoint[right](z'){$z'$}
    \tkzMarkSegments[mark= z,size=5pt](x,y x',y')
    \tkzMarkSegments[mark= s,size=5pt](y,z y',z')
    \end{tikzpicture}
    \caption{SUM7 (same figure for SUM8)}
\end{figure}

\item $\vdash (Opp(x,z,x',z') \land Bet (x,y,z)) \rightarrow \exists y' (Bet(x',y',z') \land Opp (x',y',x,y))$ \newline

The following axiom tells us that summing a null or lightlike line does not change the length: it gives back a congruent segment. 

\item $\vdash L(x,y) \, \land \, Orth(x,y,z) \, \land \, Par(x, z, y, v) \, \land Par(x,y,z,v) \, \rightarrow \\  \equiv (x,z,x,v)$

\begin{figure}[H]
    \centering
    \begin{tikzpicture}
    \tkzDefPoint(0,0){x}
    \tkzDefPoint(2,2){y}
    \tkzDefPoint(3,0){z}
    \tkzDefPoint(5,2){v}
    \tkzDrawSegment[color=yellow, thick](x,y)
    \tkzDrawSegment[color=blue, thin](x,z)
    \tkzDrawSegment[color=blue, thin](x,v)
    \tkzDrawSegment[color=blue, thin](y,v)
    \tkzDrawSegment[color=yellow, thick](z,v)
    \tkzDrawPoints(x,y,z,v)
    \tkzLabelPoint[below](x){$x$}
    \tkzLabelPoint[below](y){$y$}
    \tkzLabelPoint[below](v){$v$}
    \tkzLabelPoint[below](z){$z$}
\tkzMarkSegments[mark=s||,size=6pt](x,z x,v)
    \end{tikzpicture}
    \caption{Axiom (SUM9)}
\end{figure}
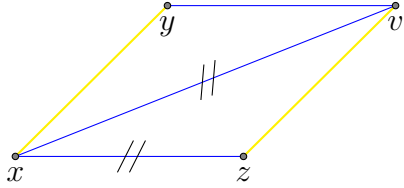

\end{enumerate}

\subsection*{Segment construction axioms}

We now want axioms that guarantee the existence of segments of a given length. They are adapted from the Euclidean context. Given two spacelike segments, we can find a third on the second line congruent to the first.

\begin{enumerate}[start=0,label={(\bfseries CONST\arabic*):}]

\item $\vdash S(x,y) \land S(z,w) \rightarrow \exists v (Coll(v,z,w) \land \equiv (z,v,x,y))$

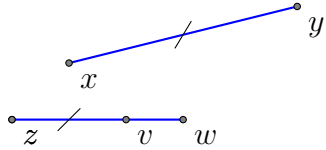
\begin{figure}[H]
\centering
\begin{tikzpicture}[scale=1.5]
\tkzInit[xmin=-1,xmax=3,ymin=-1,ymax=2]
\tkzClip
\tkzDefPoint(0,0){x}
\tkzDefPoint(2,0.5){y}
\tkzDrawSegment[color=blue,thick](x,y)
\tkzDrawPoints(x,y)
\tkzLabelPoints(x,y)
\tkzDefPoint(-0.5,-0.5){z}
\tkzDefPoint(1,-0.5){w}
\tkzDefPoint(0.5,-0.5){v}
\tkzDrawSegment[color=blue,thick](z,w)
\tkzDrawPoints(z,w)
\tkzLabelPoints(z,w)
\tkzDrawPoints(z,v)
\tkzLabelPoints(z,v)
\tkzMarkSegments[mark=s|,size=6pt](x,y z,v)
\end{tikzpicture}
\caption{Space-like segments construction.}
\end{figure}

The following two axioms guarantee that, given a spacelike and a timelike segment, we can find a third segment on the line determined by the second that is of opposite length to the first, and \textit{vice versa}.

\item $\vdash T(x,t) \land S(z,w) \rightarrow \exists v (Coll(v,z,w) \land Opp(z,v,x,t))$

\begin{figure}[H]
\centering
\begin{tikzpicture}
   \tkzDefPoint(-2,0){x}
   \tkzDefPoint(-2,2){t}
   \tkzDefPoint(0,0){z}
     \tkzDefPoint(3,0){v}
     \tkzDefPoint(2,0){w}
     \tkzDrawSegment[color=blue,thick](z,v)
         \tkzDrawSegment[color=red,thick](x,t)
         \tkzDrawPoints(x,t,z,w,v)
    \tkzLabelPoints(x,t,z,w,v)
    \tkzMarkSegments[mark=z,size=4pt](x,t z,w)
    \end{tikzpicture}
\caption{Construction of opposite segments 1}
\end{figure}
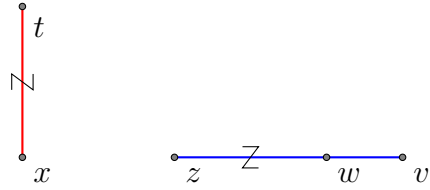

\item $\vdash T(x,t) \land S(z,w) \rightarrow \exists v (Coll(v,x,t) \land Opp(x,v,z,w))$

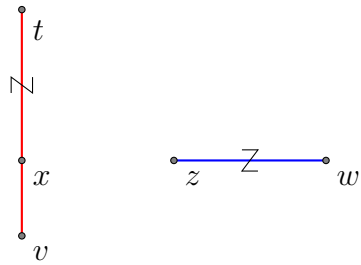
\begin{figure}[H]
\centering
\begin{tikzpicture}

    \tkzDefPoint(-2,0){x}
   \tkzDefPoint(-2,2){t}
   \tkzDefPoint(0,0){z}
     \tkzDefPoint(-2,-1){v}
     \tkzDefPoint(2,0){w}
     \tkzDrawSegment[color=blue,thick](z,w)
    \tkzDrawSegment[color=red,thick](v,t)
    \tkzDrawPoints(x,z,t,w,v)
    \tkzLabelPoints(x,t,z,w,v)
    \tkzMarkSegments[mark= z, size=4pt](x,t z,w)
    \end{tikzpicture}
\caption{Construction of opposite segments 2}
\end{figure}

The last axiom of this section postulates that a timelike line is infinite in both directions. Time has no beginning and no end. 

 \item $\vdash T(x,t) \rightarrow \exists w(Bet(w,x,t) \, \land \, \equiv (w,x,x,t)) $

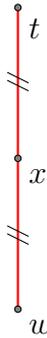
\begin{figure}[H]
\centering
\begin{tikzpicture}

    \tkzDefPoint(0, 2){t}
   \tkzDefPoint(0,-2){w}
   \tkzDefPoint(0,0){x}
     \tkzDrawSegment[color=red,thick](w,t)
    \tkzDrawPoints(x,t,w)
    \tkzLabelPoints(x,t,w)
    \tkzMarkSegments[mark=s||, size=4pt](x,t x,w)
    \end{tikzpicture}
\caption{Axiom (CONST3)}
\end{figure}

\end{enumerate}

\clearpage

\subsection*{Formal properties}

The relation of orthogonality is symmetric in the second and third term.

\begin{enumerate}[start=0,label={(\bfseries F\arabic*):}]

\item $\vdash Orth(x,y,z) \rightarrow Orth(x,z,y)$ 

The following axioms guarantee that the relation of congruence is an equivalence relation and that the relation of being shorter than induces a linear order on the equivalence classes of congruent segments.

\item $\vdash <_{\equiv} (x,y,z,w) \, \land \equiv (x,y,x',y') \rightarrow \; \; <_{\equiv} (x', y', z, w)$

\item $\vdash <_{\equiv} (x,y,z,w) \, \land \equiv (z,w,z',w') \rightarrow  \; \; <_{\equiv} (x,y,z',w') $

\item $\vdash \neg <_{\equiv} (x,y,x,y) \newline $

Degenerate segments are congruent:

\item $\vdash \neg <_{\equiv} (x,x,y,y) \newline $

These standard axioms describe the relative length between two segments that are the sum of respectively (a) congruent segments, (b) smaller segments or (c) some combination of the two. We can derive, for example, that (a) sums of congruent segments are congruent.        

\item $\vdash (Bet(x,y,z)$ $\land$ $Bet (x',y',z')$ $\land$ $\neg <_{\equiv} (x,y,x',y')$ $\land$ \newline $<_{\equiv} (y',z',y,z))$ $\rightarrow$ $<_{\equiv} (x',z',x,z)$

\item $\vdash (Bet(x,y,z)$ $\land$ $Bet (x',y',z')$ $\land$  $<_{\equiv} (x',z',x,z)$ $\land$ $\neg <_{\equiv} (x',y',x,y))$ $\rightarrow$ $<_{\equiv} (y',z',y,z)$

\end{enumerate}

\subsection*{Continuity and density}

The first axiom of continuity states that a line $\ell$ divides every plane in which it lies in two half-planes: the points whose connecting segments intersect $\ell$ and the points such that their connecting segment does not.  

\begin{enumerate}[start=0,label={(INT):}]

\item$ \vdash (Copl(x,y,z,w) \land Copl(t,x,z,w) \land \neg Intersect(x,y,z,w) \land Intersect (y,t,z,w)) \rightarrow Intersect (x,t,z,w)$

\begin{figure}[H]
\centering
\begin{tikzpicture}

    \tkzDefPoint(0,0){z}
   \tkzDefPoint(3,0){w}
   \tkzDefPoint(0.5,1){x}
   \tkzDefPoint(2.4,0.7){y}
   \tkzDefPoint(1.6,-1.5){t}
    \tkzDrawSegment(z,w)
    \tkzDrawSegment(x,y)
    \tkzDrawSegment(y,t)
    \tkzDrawSegment(t,x)
    \tkzDrawPoints(x,y,z,w,t)
    \tkzLabelPoint[below](z){$z$}
    \tkzLabelPoint[below](w){$w$}
    \tkzLabelPoint[above](x){$x$}
    \tkzLabelPoint[above](y){$y$}
    \tkzLabelPoint[below](t){$t$}
    \end{tikzpicture}
\caption{Axiom (INT)}
\end{figure}
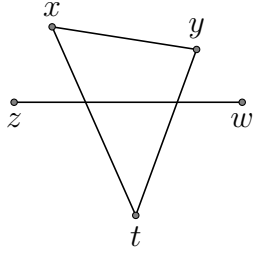

\end{enumerate}

These axioms are imported from [Tarski and Givant, 1999]. The continuity schema constrains the ordering of the points on a line to be as Dedekind complete as possible, without quantifying over sets of points. Density is the usual fact that between every two distinct points there is a third.

\begin{enumerate}[start=0,label={(ASC):}]

\item $\vdash \exists x \, \forall y \, \forall z (\phi \land \psi \rightarrow Bet(x,y,z)) \rightarrow \exists x' \, \forall y \, \forall z (\phi \land \psi \rightarrow Bet(x',y,z))$  

where $\phi$ and $\psi$ are formulae of $L$, the first of which does not contain any free occurrences of $x$, $x',$ $z$, the second of which does not contain any free occurrences of $x$, $x',$ $y$.

\end{enumerate}
\begin{enumerate}[start=0,label={(DENS):}]

\item $\vdash x \neq z \rightarrow \exists y( y \neq x \land y \neq z \land Bet (x,y,z))$  \\

\begin{figure}[H]
\centering
\begin{tikzpicture}

    \tkzDefPoint(0,0){x}
   \tkzDefPoint(1.5,0){y}
   \tkzDefPoint(3,0){z}
    \tkzDrawSegment(x,z)
    \tkzDrawPoints(x,y,z))
    \tkzLabelPoints(x,y,z)
    \end{tikzpicture}
\caption{Axiom (DENS)}
\end{figure}
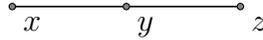

\end{enumerate}

This completes the presentation of our system for Minkowski spacetime, which will be denoted by $\mathcal{M}^1$. Its adequacy can now be briefly investigated. 

\subsection*{A Second-order Continuity Axiom}

Minkowski spacetime is the `intended' model of the system $\mathcal{M}^1$. It is the physical spacetime that is postulated by the theory of Special Relativity (SR). It can be singled out, up to isomorphism, as an uncountable model $\mathfrak{M}$ of $\mathcal{M}^1$ such that the lines $\ell$ in $\mathfrak{M}$ are true continua. In a line $\ell$ in $\mathfrak{M}$ every bounded set of points has a least upper bound. An equivalent method is to look at models of the following second order continuity axiom (ASC):

\begin{enumerate}[start=0,label={(CONT):}]

\item $\forall X$ $\forall Y$ $\exists x$ $\, \forall y$ $\, \forall z$  $(X(x) \land Y(y) \rightarrow Bet(x,y,z)) \rightarrow$ $\exists w$ $\forall y$ $\forall z$ $(X(x)$ $\land$ $Y(y) \rightarrow Bet(w, y, z))$ 
\end{enumerate}

The system we obtain from $\mathcal{M}^1$ by replacing all instances of (ASC) with (CONT) will be denoted by $\mathcal{M}^2$ and we will now consider its models. 

\section{Representation Theorems}

In a physics textbook, `Minkowski Spacetime' refers to a certain mathematical structure $\langle  \mathbb{R}^{4},\eta_{ab} \rangle$. It is assumed that calculations performed on $\langle  \mathbb{R}^{4},\eta_{ab} \rangle$ reflect certain physical state of affairs in the physical manifold of points in which physical objects are located, and on which physical fields assume values. Let us call Minkowski spacetime$_2$ a spacetime obeying the axioms of $\mathcal{M}^1$ and Minkowski spacetime$_1$ the following mathematical structure:      

\begin{definition}
The tensor $\eta_{ab}$ is the covariant tensor on $\mathbb{R}^{4}$ such that, for all $p, q$ $\epsilon$ $\mathbb{R}^{4}$, $\eta_{ab} (p,q) =$ $\sqrt{(t_p-t_q)^2 - (x_p- x_q)^2 - (y_p- y_q)^2 - (z_p - z_q)^2}$
\end{definition} 

When $p, q$ $\epsilon$ $\mathbb{R}^{4}$, we have denoted by $x_p$ the first number in the quadruple $p$, by $y_q$ the second number in the quadruple $q$ and so on.   

\begin{definition}
The function distance$_{4}$ is the tensor on $\mathbb{R}^{4}$ such that, for all $p, q$ $\epsilon$ $\mathbb{R}^{4}$, distance$_{4}$($p,q$) $=$ $\sqrt{(x_p- x_q)^2 + (y_p- y_q)^2 + (z_p - z_q)^2 + (t_p-t_q)^2}$     
\end{definition} 

\begin{definition}
Minkowski spacetime$_1$ is $\langle  \mathbb{R}^{4}, \eta_{ab}$, distance$_{4}\rangle$  
\end{definition}

The sense in which $\langle  \mathbb{R}^{4}, \eta_{ab}$, distance$_{4}\rangle$  can be used used to `represent' Minkowski spacetime$_2$, and the role of frames of reference, is clarified by proving a \textit{Representation Theorem}. A frame or coordinatization is a bijection  $f$: Minkowski spacetime$_2 \rightarrow$ Minkowski spacetime$_1$ such that spacetime events satisfy intrinsic geometry relations if and only if their images satisfy corresponding algebraic relationships. The passive symmetries of the theory $\mathcal{M}^2$ emerge as the transformations $f$: $\langle  \mathbb{R}^{4},\eta_{ab}\rangle$ $\rightarrow \langle  \mathbb{R}^{4},\eta_{ab} \rangle$ that can be composed with an arbitrary coordinatization to leave a coordinatization.  

\subsection{Classical Representation Theorems}

Tarski and his students have constructed a simple system of axioms $\mathcal{E}_3$ for Euclidean geometry in three dimensions. We have already mentionned it and we have exploited it in the formulation of our axioms. A model $\mathfrak{F}$ of $\mathcal{E}_3$ is represented by the mathematical structure $\langle \mathbb{R}^{3}$, distance$_{s} \rangle$ :

\begin{definition}
The function distance$_{s}$ is the function on $\mathbb{R}^{3}$ such that, for all $p, q$ $\epsilon$ $\mathbb{R}^{3}$, distance$_{s}$($p,q$) $=$ $\sqrt{(x_p- x_q)^2 + (y_p- y_q)^2 + (z_p - z_q)^2}$     
\end{definition}

An analogue system of axioms $\mathcal{E}_1$ for one dimensional temporal geometry can be constructed. A model $\mathfrak{F}$ of $\mathcal{E}_1$ is mirrored by the real line, that is by the mathematical structure $\langle  \mathbb{R}$, distance$_{t} \rangle$, where: 

\begin{definition}
The function distance$_{t}$ is the tensor on $\mathbb{R}$ such that, for all $p, q$ $\epsilon$ $\mathbb{R}$, distance$_{t}$($p,q$) $=$ $\sqrt{(t_p-t_q)^2}$ $=$ $|t_p-t_q|$
\end{definition}  

The proof of our \textit{Representation Theorem}, relating the system $\mathcal{M}^1$ to the structure $\langle \mathbb{R}^{4}, \eta_{ab}$, distance$_{4}\rangle$, will follow from two theorems of [Tarski, 1959] and a theorem of [Suppes, 1959]. The first theorem of [Tarski, 1959] is simply the appropriate \textit{Representation Theorem} for his own axiom system $\mathcal{E}_3$ for Eucliden geometry and for the Cartesian space $\langle \mathbb{R}^{3}$, distance$_{s} \rangle$.   

\begin{theorem}
(Tarski 1959). $\mathfrak{M}$ is a model of $\mathcal{E}_3$ if and only if there is a bijection $f$: U($\mathfrak{M}) \rightarrow \mathbb{R}^{3}$ such that, for all choices of a, b, c, d $\epsilon$ U($\mathfrak{M}$):\newline

(1) $\mathfrak{M}$ $\models$ $Bet(x,y,z)$ $[a,b,c$] \space if and only if  \space distance$_{s}$($f(a),f(c)$) = distance$_{s}$($f(a),f(b)$) $+$ distance$_{s}$($f(b),f(c)$) \newline  

(2) $\mathfrak{M}$ $\models$ $\ulcorner \equiv(x,y,z,w) \urcorner$ $[a,b,c,d$] \space if and only if  \space distance$_{s}$($f(a),f(b)$) = distance$_{s}$($f(c),f(d)$)  \newline 

and, for every two functions $f$ and $f'$ that satisfy (1)-(2), there exists an isometry I: $\mathbb{R}^{3} \rightarrow \mathbb{R}^{3}$ and a function U: $\mathbb{R}^{3} \rightarrow \mathbb{R}^{3}$ that multiplies each entry by a fixed constant such that $f$ = $U \circ I \circ f'$.    
\end{theorem}

A similar theorem can be proven for an appropriate system of one dimensional geometry [cf. Tarski and Givant 1999, pp. 204-209].  

\begin{theorem}
(Tarski 1959). $\mathfrak{M}$ is a model of $\mathcal{E}_1$ if and only if there is a bijection $f$: U($\mathfrak{M}) \rightarrow \mathbb{R}$ such that, for all choices of a, b, c, d $\epsilon$ U($\mathfrak{M}$):\newline

(1) $\mathfrak{M}$ $\models$ $\ulcorner \equiv(x,y,z,w) \urcorner$ $[a,b,c,d$] \space if and only if  \space distance$_{t}$($f(a),f(b)$) = distance$_{t}$($f(c),f(d)$) \newline 

and, for every two functions $f$ and $f'$ that satisfy (1)-(2), there exists a translation $b$: $\mathbb{R} \rightarrow \mathbb{R}$ and a function $k$: $\mathbb{R} \rightarrow \mathbb{R}$ that consists of multiplying all components by a constant, such that $f$ = $k \circ b \circ f'$.    
\end{theorem}

The theorem of [Suppes, 1959] is a basic result\footnote{Suppes [1959] proves in fact a stronger results. He only assumes that $f$ and $f'$ agree on lightlike and timelike connected points. Note that Suppes [1959] refers to the Poincaré transformations - the composition $b \circ L$ of a translation $b$ and a linear transformation $L$ corresponding to a Lorentz matrix - as the `Lorentz transformations'.}  characterizing the relation between the relativistic intervals and the Poincaré transformations on $\langle  \mathbb{R}^{4}, \eta_{ab}\rangle$:

\begin{theorem}
(Suppes 1959). For any two bijective functions $f$: $M \rightarrow \mathbb{R}^{4}$ and $f'$: $M \rightarrow \mathbb{R}^{4}$ from the same uncountable set $M$ into $\mathbb{R}^{4}$ such that, for all $p, q$ $\epsilon$ $M$, $\eta_{ab} (f(p),f(q))$ = $\eta_{ab} (f'(p),f'(q))$ .i.e, they agree on the relativistic interval, there exists a Poincaré transformation $L$ such that $f$ = $L \circ f'$.
\end{theorem}

\subsection{A Representation Theorem for Minkowski Spacetime}

The following theorem is the main result of this paper: 

\begin{theorem}
(\textbf{Representation theorem}). $\mathfrak{M}$ is a model of $\mathcal{M}^2$ if and only if there is a bijection $f$: U($\mathfrak{M}) \rightarrow \mathbb{R}^{4}$ such that, for all a, b, c, d $\epsilon$ U($\mathfrak{M}$):\newline

(1) $\mathfrak{M}$ $\models$ $Bet(x,y,z)$ $[a,b,c$] \space if and only if  \space distance$_{4}$($f(a),f(c)$) = distance$_{4}$($f(a),f(b)$) $+$ distance$_{4}$($f(b),f(c)$) \newline  

(2) $\mathfrak{M}$ $\models$ $\ulcorner <_{\equiv}(x,y,z,w) \urcorner$ $[a,b,c,d$] if and only if $\eta_{ab}^2(f(a),f(b)$) $\leq$ $\eta_{ab}^2(f(c),f(d)$)  \newline 

and, for every two functions $f$ and $f'$ that satisfy (1)-(2), there exists a Poincaré transformation L: $\mathbb{R}^{4} \rightarrow \mathbb{R}^{4}$, and a function U: $\mathbb{R}^{4} \rightarrow \mathbb{R}^{4}$ that multiplies each coordinate by a positive constant, such that $f$ = $L \circ U \circ f'$.\footnote{Field [1980, p. 50/f] has noticed the addition of $U$ to the group of symmetries is due to the conventionality of the choice of measuring units. It  marks the difference between e.g., measuring the relativistic interval in second, minutes or hours. }     
\end{theorem}

The main idea behind the proof of the existence part is to start from a basis with a given time axis $L$ and and a spacelike hyperplane $E$, and then extend coordinatizations $f$ and $g$ of $E$ and of $L$, given by \textit{Theorem 1.} and \textit{Theorem 2.}, to a coordinatization $f'$ of the entirety of $\mathfrak{M}$. The specification of how to extend $f$ and $g$ can be done in a uniform way. In all the definitions that follow, let $\mathfrak{M}$ be a model of $\mathcal{M}^2$ and let $o$, $x$, $y$, $z$, and $t$ determine a basis in the model $\mathfrak{M}$. Let $E$ be the associated spacelike hyperplane in $\mathfrak{M}$ and L be the timelike line through $o$ and $t$.\footnote{$E$ is the set of elements of the domain $U(\mathfrak{M})$ that are generated by $o, x, y, z$ in $\mathfrak{M}$ and $L$ is the set of  elements of the domain $U(\mathfrak{M})$ that collinear to $o, t$ in $\mathfrak{M}$.} We will use subscripts to denote components. For example, if f($p$)= $\langle 2, 4,15 \rangle$, then f$_3$($p$)= 15 and f$_1$($p$)= 2. 

\begin{definition}
Let us assume that $f$: $E \rightarrow \mathbb{R}^{3}$ satisfies condition \textit{(1)-(2)} in \textit{Theorem 1} when restricted to E and that $g$: $L \rightarrow \mathbb{R}$ satisfies \textit{(1)-(2)} in \textit{Theorem 2} when restricted to L. Assume $\eta_{ab}(g(o),g(t)$) = - $\eta_{ab}(f(o),f(x)$). A function $f'$: U($\mathfrak{M}) \rightarrow \mathbb{R}^{4}$ is \textit{determined} by $f$ and $g$ if and only if:
\end{definition} 

\begin{enumerate}
\item $f'$ ($p$)= $\langle$f$_1$($p$), f$_2$($p$), f$_3$($p$), 0$\rangle$ \space \space \space if $p$ $\epsilon$ $E$
\item $f'$ ($p$)= $\langle$0, 0, 0, g$_1$($p$)$\rangle$ \space \space \space if $p$ $\epsilon$ $L$
\item $f'$($p$)= $\langle$f$_1$($q$), f$_2$($q$), f$_3$($q$), g$_1$($t$) $\rangle$ if $\vec{op}$ is the sum of $\vec{ot}$ $\epsilon$ $L$ and $\vec{oq}$ $\epsilon$ $E$.\footnote{If $Reached^{\mathfrak{M}}(o,q,t,q,p)$ and the segment $\vec{tq}$ does not intersect the hyperplane $E$}     
\end{enumerate}

\textbf{Remark on notation}: We follow the conventions of Shoenfield [1971] for models of set theory and use superscripts to form  predicates for the satisfaction of object language predicates in a model. For example, we have $Orth^{\mathfrak{M}}(b,a,c)$ if and only if $a$, $b$, $c$ $\epsilon$ U($\mathfrak{M}$) and $a$, $b$,  $c$ are orthogonal in $\mathfrak{M}$. \newline 

The first preliminary lemma tells us that in $\mathfrak{M}$ a quadrilateral with two right angles at the base and the other two sides parallel and congruent is a rectangle: opposite sides are congruent and all angles are right.  

\begin{figure}[H]
\centering
\begin{tikzpicture}

    \tkzDefPoint(0,0){a}
   \tkzDefPoint(1.3,0){b}
   \tkzDefPoint(0,2){c}
   \tkzDefPoint(1.3,2){d}
   \tkzDefPoint(1.3,-2){d'}
    \tkzDrawSegment(a,b)
    \tkzDrawSegment(a,c)
    \tkzDrawSegment(c,d)
    \tkzDrawSegment(b,d)
    \tkzDrawSegment[dashed](b,d')
    \tkzDrawPoints(a,b,c,d,d')
    \tkzLabelPoint[left](a){$a$}
    \tkzLabelPoint[right](b){$b$}
    \tkzLabelPoint[left](c){$c$}
    \tkzLabelPoint[right](d){$d=e$}
    \tkzLabelPoint[right](d'){$\overline{d}$}
     \tkzMarkSegments[mark=s|, size=4pt](a,c b,d)
     \tkzMarkRightAngle(b,a,c)
     \tkzMarkRightAngle(a,b,d)
    \end{tikzpicture}
\caption{Lemma 1}
\end{figure}

\begin{Lemma}
 For any $a$, $b$, $c$,  $d$ $\epsilon$ U($\mathfrak{M}$), if $Orth^{\mathfrak{M}}(b,a,d)$,   $Orth^{\mathfrak{M}}(a,b,c)$, $\equiv^{\mathfrak{M}}(a,c,b,d)$ and $Par^{\mathfrak{M}}(a,c,b,d)$, then $\equiv^{\mathfrak{M}}(a,b,c,d)$ and $Orth^{\mathfrak{M}}(d,b,c)$.
\end{Lemma}

\begin{proof}
 $Orth^{\mathfrak{M}}(a,b,c)$ and $Par^{\mathfrak{M}}(a,c,b,d)$ by hypothesis. It follows by definition that $Reached^{\mathfrak{M}}(b,a,d,a,c)$. By axiom (A5), there exists an $e$ in U($\mathfrak{M}$) such that $Reached^{\mathfrak{M}}(b,d,a,e,c)$ and we have the two congruences $\equiv^{\mathfrak{M}}(e,c,b,a)$ and $\equiv^{\mathfrak{M}}(b,e,a,c)$. The hypothesis $\equiv^{\mathfrak{M}}(a,b,c,d)$ and the transitivity of congruence (F1)-(F3) imply that $\equiv^{\mathfrak{M}}(b,e,b,d)$. By definition again, the fact that $Reached^{\mathfrak{M}}(b,d,a,e,c)$ implies that $Coll^{\mathfrak{M}}(b,d,e)$ and $Orth^{\mathfrak{M}}(e,b,c)$. Axioms (SUM4) to (SUM5) and the fact that $\equiv^{\mathfrak{M}}(b,e,b,d)$, reduce now the choice to either $e=d$ or $e=\overline{d}$, where $\overline{d}$ is the reflection of $d$ over the line $\ell$ through $a$ and $b$. But Axiom (INT) excludes that $e=\overline{d}$ . So  $e=d$, and therefore we have $\equiv^{\mathfrak{M}}(a,b,c,d)$ and $Orth^{\mathfrak{M}}(d,b,c)$  .
\end{proof}

\subsubsection{Lemmata on Opposites}

The coordinatizations $f$ and $g$ are worth combining together only if $\eta_{ab}^2(g(o),g(t)$) = - $\eta_{ab}^2(f(o),f(x)$). This obviously implies that $\eta_{ab}^2(f'(o),f'(t)$) = - $\eta_{ab}^2(f'(o),f'(x)$). In general, two segments $\vec{ot'}$ ($t'$ $\epsilon$ $L$) and  $\vec{ow}$ ($w$ $\epsilon$ $E$)  are of opposite length in $\mathfrak{M}$ if and only if $\eta_{ab}^2(f'(o),f'(t')$) = - $\eta_{ab}^2(f'(o),f(w)$).  

\begin{Lemma}
For any $t$, $t'$ $\epsilon$ $L$ and $x$, $x'$ $\epsilon$ $E$, if  $Bet^{\mathfrak{M}}(o,t,t')$ and $Bet^{\mathfrak{M}}(o,x,x')$,  $\eta_{ab}^2(f'(o),f'(t)$) = - $\eta_{ab}^2(f'(o),f(x))$ and $\eta_{ab}^2(f'(t),f'(t')$) = - $\eta_{ab}^2(f'(x),f(x'))$, then we have that $\eta_{ab}^2(f'(o),f'(t')$) = - $\eta_{ab}^2(f'(o),f(x'))$.
\end{Lemma}
\begin{proof}

\begin{align*}
\eta_{ab}(f'(o),f'(t')) =  |g(t') - g(o)| &\qquad  \text{(By definition 8 and $t'$ $\epsilon$ $L$)} \\
=  |(g(t') - g(t)) + (g(t) - g(o))|  &\qquad \text{} \\
=  | g(t') - g(t)|  + |g(t) - g(o)|  &\qquad \text{(By $Bet^{\mathfrak{M}}(o,t,t')$ and condition 1 of Theorem 2)}\\
= \eta_{ab}(f'(t),f'(t')) + \eta_{ab}(f'(o),f'(t))  &\qquad \text{(By definition 8)}\\
= i \eta_{ab}(f'(o),f(x)) + i \eta_{ab}(f'(x),f(x'))  &\qquad \text{(By hypothesis)}
\end{align*}

An analogous argument shows that:
\begin{align*}
\eta_{ab}(f'(o),f'(x')) = \eta_{ab}(f'(o),f(x)) + \eta_{ab}(f'(x),f(x')). 
\end{align*}
\end{proof}

Let us now prove the existence of a rectangle with two given sides.

\begin{Lemma}
For all $o$, $t$, $x$ in U($\mathfrak{M}$), if $Orth^{\mathfrak{M}}(o,t,x)$ then there exists a $w$ in U($\mathfrak{M}$) such that $Reached^{\mathfrak{M}}(o,x,t,x,w)$ and  $Reached^{\mathfrak{M}}(o,t,x,t,w)$.
\end{Lemma}

\begin{proof}
By axiom (A9) there exists a $w$ in U($\mathfrak{M}$) such that $Reached^{\mathfrak{M}}(o,x,t,x,w)$. Axioms (CONST3), together with the definitions of orthogonality and parallelism, implies that also the reflection $\overline{w}$ over the line $\ell$ through $\vec{ox}$ is such that $Reached^{\mathfrak{M}}(o,x,t,x,\overline{w})$. (SUM6)(SUM7)(SUM8) imply that there are no others.  Similarly $t$ and $\overline{t}$ are the only points $t'$ in U($\mathfrak{M}$) such that Coll$^{\mathfrak{M}}(o,t,t'$) and $\equiv^{\mathfrak{M}}(o,t,o,t')$. Either $Reached^{\mathfrak{M}}(o,t,x,t,w)$ or $Reached^{\mathfrak{M}}(o,t,x, \overline{t},w)$. In the second case, we get that $Reached^{\mathfrak{M}}(o,t,x,t,\overline{w})$ by (INT) and (A6).     
\end{proof}

The sums of segments opposite length are of opposite length.

\begin{Lemma}
For every $t$, $t'$, $x$, $x'$, if $Bet^{\mathfrak{M}}(o,x,x')$, $Bet^{\mathfrak{M}}(o,t,t')$, $Opp^{\mathfrak{M}}(o,x,o,t)$ and $Opp^{\mathfrak{M}}(t,t',x,x')$, then we have that $Opp^{\mathfrak{M}}(o,x',o,t')$. 
\end{Lemma}

\begin{proof}
Lemma 3 gives us $r$ and $r'$ in U($\mathfrak{M}$) such that $Reached^{\mathfrak{M}}(o,x,t,x,r)$ and $Reached^{\mathfrak{M}}(o,x',t',x',r')$ and also the alternative developments: that is $Reached^{\mathfrak{M}}(o,t,x,t,r)$ and $Reached^{\mathfrak{M}}(o,t',x',t',r')$. Let $v$ and $v'$ be such that similarly $Reached^{\mathfrak{M}}(o,x,t',x,v)$, and $Reached^{\mathfrak{M}}(o,x,t,x',v')$, and analogous permutations. Various applications of Lemma 1 to all the different rectangles in Fig.28 entail that $Reached^{\mathfrak{M}}(r,v,v',v',r')$ and the two congruences: $\equiv^{\mathfrak{M}}(r,v,x,x')$ and $\equiv^{\mathfrak{M}}(r,v',t,t')$. Axiom (SUM 3) implies that $L^{\mathfrak{M}}(r,r')$.  By continuity, for all choices of unit $e$, there must be some segment $\vec{ww'}$ such that Product($o,e, o,x,w,w',o,x')$. The definition of $Product''(o,e,o,t,w,w',o,t')$ and the hypotheses $Opp^{\mathfrak{M}}(o,x,o,t)$ and $Opp^{\mathfrak{M}}(t,t',x,x')$ imply that all the conditions in axiom (A9) are satisfied. This means that $Bet^{\mathfrak{M}}(o,r,r')$. By the degenerate cases of axioms (F4) and (F5), it follows that $L^{\mathfrak{M}}(o,r')$. By definition of $Opp$, we obtain immediately that $Opp^{\mathfrak{M}}(o,t',o,x')$. 
\end{proof}

\begin{figure}[H]
    \centering
    \begin{tikzpicture}
    \tkzDefPoint(0,0){o}
    \tkzDefPoint(3,0){x'}
    \tkzDefPoint(1.5,0){x}
    \tkzDefPoint(0,2){t'}
    \tkzDefPoint(0,1){t}
    \tkzDefPoint(3,2){r'}
    \tkzDefPoint(1.5,1){r}
    \tkzDefPoint(1.5,2){v}
    \tkzDefPoint(3,1){v'}
     \tkzDefPoint(1.3,1){u}
    \tkzDrawSegment[dashed](r',x')
    \tkzDrawSegment[dashed](r,x)
    \tkzDrawSegment[dashed](t,r)
    \tkzDrawSegment[dashed](t',r')
    \tkzDrawSegment[color=blue](o,x')
    \tkzDrawSegment[color=yellow](r,r')
    \tkzDrawSegment[color=red](o,t')
    \tkzDrawSegment[color=red, dashed](r,v)
    \tkzDrawSegment[color=blue, dashed](r,v')
    \tkzDrawPoints(o,x,x',t,t',r,r',v,v')
    \tkzLabelPoint[below](o){$o$}
    \tkzLabelPoint[below](x){$x$}
    \tkzLabelPoint[right](x'){$x'$}
    \tkzLabelPoint[left](t'){$t'$}
    \tkzLabelPoint[left](t){$t$}
    \tkzLabelPoint[above](u){$r$}
    \tkzLabelPoint[right](r'){$r'$}
    \tkzLabelPoint[above](v){$v$}
    \tkzLabelPoint[right](v'){$v'$}
    \end{tikzpicture}
    \caption{Lemma 4}
\end{figure}

\begin{Lemma}
For every $t$, $t'$, $x$, $x'$ $\epsilon$ $U(\mathfrak{M})$, if $Bet^{\mathfrak{M}}(o,x,x')$, $Bet^{\mathfrak{M}}(o,t,t')$, $Opp^{\mathfrak{M}}(o,x,o,t)$ and $Opp^{\mathfrak{M}}(o,x',o,t')$, then $Opp^{\mathfrak{M}}(t,t',x,x')$. 
\end{Lemma}

\begin{proof}
The proof is similar to that of Lemma 4. 
\end{proof}

\begin{Lemma}
For any $o$, $t'$ $\epsilon$ $L$ and $w$ $\epsilon$ $E$, $Opp^{\mathfrak{M}}(o,t',o,w)$ if and only if $\eta_{ab}^2(f'(o),f'(t')$) = - $\eta_{ab}^2(f'(o),f'(w)$).  
\end{Lemma}

\begin{proof} The hypothesis is that $\eta_{ab}^2(g(o),g(t)$) = - $\eta_{ab}^2(f(o),f(x)$). Lemma 4 and Lemma 2 imply by induction that the statement holds for all integer multiples of the above segments $n \cdot \vec{ox}$ and $n \cdot \vec{ot}$. Lemma 5 and Lemma 2 imply the result for integer submultiples $\frac{1}{n} \cdot \vec{ox}$ and $\frac{1}{n} \cdot \vec{ot}$. By continuity, for all reals $k$, we have that $k \cdot \vec{ox}$ and $k \cdot \vec{ot}$ are opposites\footnote{If $\neg Opp^{\mathfrak{M}}(o, k \cdot t, o, (-k) \cdot x $, it will follow from (SUM 1) and (SUM 2) and the continuity of lines that there is some real $r'$ such that $Opp^{\mathfrak{M}}(o, k \cdot t, o, -r' \cdot x)$ or $Opp^{\mathfrak{M}}(o, r' \cdot t, o, -k'\cdot x)$. It suffices, then, to pick a rational $\frac{p}{q}$ such that $r' < \frac{p}{q} < r$ and notice that corresponding multiples of the segments are of opposite length and in-between two opposite irrational segments. This contradicts basic consequences of axioms (SUM0) and (SUM6)(SUM8) about the ordering of opposites.} in $\mathfrak{M}$ and the identity $\eta_{ab}^2(f'(o),f'(k \cdot t)$) = - $\eta_{ab}^2(f'(o),f(k \cdot x)$). Axiom (SUM0) states that segments  of opposites length to a given segment are congruent. Axiom (SUM6)(SUM7)(SUM8) imply that congruent segments on the lines $\ell_{o,x}$ and $\ell_{o,t}$ are of the form $k \cdot \vec{ox}$ and $(-k) \cdot   \vec{ox}$, or of the form $k \cdot \vec{ot}$ and $(-k) \cdot \vec{ot}$. Euclidean geometry (A3) and Theorem 1 imply that every segment in $E$ is congruent to and of same interval (relative to $f'$) as a segment in $\ell_{o,x}$.             
\end{proof}

\subsubsection{Lemmata on the \textit{Streckenrechnung}}

The lemmata in this section consist merely in a verification of the adequacy of the `calculus of segments' of Hilbert [1899]. 

\begin{Lemma}
If $\phi$ is a formula of the calculus of segments (D20)-(D29), for any $o,e,v_1,...,v_{2n}$,$v_1',...,v_{2n}'$ $\in$ U($\mathfrak{M}$) such that (1) for all $i< 2n$, $\equiv^{\mathfrak{M}} (v_i, v_{i+1}, v_i', v_{i+1}')$ and (2) $\phi^{\mathfrak{M}}(o,e, v_1,...,v_{2n})$ and (3) $\phi^{\mathfrak{M}}(o,e, v_1',...,v_{2n}')$, $\equiv^{\mathfrak{M}} (v_{2n-1}, v_{2n},v_{2n-1}',v_{2n}')$.
\end{Lemma}

\begin{Lemma}
For all $o',e,x,y,z,w,v,l$ $\in $ $E$, if $\eta_{ab}(f'(o'),f'(e)) = 1$, then  $Product^{\mathfrak{M}}(o',e,o,x,w,y,v,z)$ iff $\eta_{ab}(f'(v),f'(z))$ $=$ $\eta_{ab}(f'(o),f'(x)) \eta_{ab}(f'(w),f'(y))$.
\end{Lemma}

 Let us fix two points $o'$ and $e$ $\epsilon$ $E$ such that $\eta_{ab}(f'(o'),f'(e)) = 1$. 

\begin{Lemma}
For all $x,y,z,w,v$ $\in $ $E$, $Remterm^{\mathfrak{M}}(o',e, o,x,w,y,v,z)$ if and only if $\eta_{ab}^2(f'(v),f'(z))$ $=$ $\eta_{ab}^2(f'(o),f'(x)) + 2\eta_{ab}(f'(o),f'(x))\eta_{ab}(f'(w),f'(y))$.
\end{Lemma}

\begin{Lemma}
For all $x,y,z,w,v$ $\in $ $L$, $Remterm'^{\mathfrak{M}}(o',e, o,x,w,y,v,z)$ if and only if $\eta_{ab}^2(f'(v),f'(z))$ $=$ $\eta_{ab}^2(f'(o),f'(x)) + 2\eta_{ab}(f'(o),f'(x))\eta_{ab}(f'(w),f'(y))$.
\end{Lemma}

\begin{proof}
These results can be derived from the theory of proportions in an Euclidean space [Hartshorne 2000, Schwabhäuser, Szmielew and Tarski 1983] and details are omitted. Note that our formulation of (A4) allows us to apply the usual congruence criteria for triangles across different hyperplanes. 
\end{proof}

\subsubsection{Lemmata on transport to the origin}

A basic property of the model $\mathfrak{M}$ is that pairs of segments that decompose into congruent components on an orthogonal basis are congruent.   

\begin{Lemma}
For all $o,x,t,v,o',x',t',v'$ in U($\mathfrak{M}$), if $Reached^\mathfrak{M}(o,x,t,x,v)$, $Reached^\mathfrak{M}(o',x',t',x',v')$, $\equiv^\mathfrak{M}(o,x,o',x')$, $\equiv^\mathfrak{M}(o,t,o',t')$, then $\equiv^\mathfrak{M}(o,v,o',v')$. 
\end{Lemma}

\begin{proof}
This is proven by cases. If $L^\mathfrak{M}(o,v)$, then by definition $Opp^\mathfrak{M}(o,x,o,t)$. By axiom (SUM6), it follows that $Opp^\mathfrak{M}(o',x',o',t')$. By axiom (SUM3), we have that $L^\mathfrak{M}(o',v')$. By axiom (F1) to (F4), every two lightlike segments are congruent. If $S^\mathfrak{M}(o,v)$ or $T^\mathfrak{M}(o,v)$, the result follows from axioms (SUM1) and (SUM2) and Lemma 6 on the Streckenrechnung. 
\end{proof}

Our formal verification that the function $f$ in \textit{Definition 8} satisfies conditions \textit{(1)} and \textit{(2)} of \textit{Theorem 4} requires that we be able to restrict ourselves to the case of segments $\vec{ox}$ and $\vec{oy}$ stemming from the same origin $o$. The next lemma shows that to an arbitrary segment $\vec{pq}$ we can associate a congruent vector $\vec{or}$ at the origin such that $f$ assigns to them the same interval. 

\begin{Lemma}
For all $p$, $q$ in U($\mathfrak{M}$), there is an $r$ $\epsilon$ $U(\mathfrak{M})$ such that \newline  $\equiv^\mathfrak{M}(p,q,o,r)$ and  $\eta_{ab}^2(f'(p),f'(q)) = \eta_{ab}^2(f'(o),f'(r))$. 
\end{Lemma}

\begin{proof}
Let $t,x,t',x'$ be the points such that $Reached^\mathfrak{M}(o,x,t,x,p)$ and $Reached^\mathfrak{M}(o,x',t',x',q)$, as guaranteed by Lemma 4. Let $t''$ and $w''$ be the points: \vspace{0.1 cm}

$$ t''=: f'^{-1}(f'(t')-f'(t)) $$
$$ x''=: f'^{-1}(f'(x')-f'(x)) $$ \vspace{0.1 cm}

Let $r$ be such that $Reached^\mathfrak{M}(o,x'',t'',x'',r)$ and  $Reached^\mathfrak{M}(o,t'',x'',t'',r)$. The identity $\eta_{ab}^2(f'(p),f'(q)) = \eta_{ab}^2(f'(o),f'(r))$ is obvious. \newline

Theorem 1 gives us $\equiv^\mathfrak{M}(x,x',o,x'')$. Theorem 2 implies that $\equiv^\mathfrak{M}(t,t',o,t'')$). Let $v$ be the point such that $Reached^\mathfrak{M}(o,x',t,x',v)$ and and $v'$ be the point such that reached $Reached^\mathfrak{M}(o,x,t',x,v')$. Applications of Lemma 1 to the different rectangles in Fig. 30 establish that $\equiv^\mathfrak{M}(p,v,x,x')$ and that  $\equiv^\mathfrak{M}(p,v',t,t')$. We also get that $Orth^\mathfrak{M}(p,v',v)$. The transitivity of congruence imply that $\equiv^\mathfrak{M}(p,v,o,x''$ and  $\equiv^\mathfrak{M}(p,v',o,t'')$). By definition $Reached^\mathfrak{M}(p,v,v',v,q)$. By the preceding lemma, we obtain that $\equiv^\mathfrak{M}(o,r,p,q)$. 
\end{proof}

\begin{figure}[H]
    \centering
    \begin{tikzpicture}
    \tkzDefPoint(0,0){o}
    \tkzDefPoint(1,0){w''}
    \tkzDefPoint(3,0){w}
    \tkzDefPoint(4,0){w'}
    \tkzDefPoint(0,0.7){z''}
    \tkzDefPoint(0,1.3){z}
    \tkzDefPoint(0,2){z'}
    \tkzDefPoint(3,1.3){p}
    \tkzDefPoint(4,2){q}
    \tkzDefPoint(1,0.7){r}
    \tkzDefPoint(4,1.3){v}
    \tkzDefPoint(3,2){v'}
    \tkzDrawSegment[color=blue](o,w')
    \tkzDrawSegment[color=red](o,z')
    \tkzDrawSegment(o,r)
    \tkzDrawSegment(p,q)
    \tkzDrawSegment[dashed, color=blue](p,v)
    \tkzDrawSegment[dashed, color=red](p,v')
    \tkzDrawSegment[dashed](w'',r)
    \tkzDrawSegment[dashed](z'',r)
    \tkzDrawSegment[dashed](z,p)
    \tkzDrawSegment[dashed](w,p)
    \tkzDrawSegment[dashed](w',q)
    \tkzDrawSegment[dashed](z',q)
    \tkzDrawPoints(o,w,w'',w',z',z,z'',p,q,r,v,v')
    \tkzLabelPoint[below](o){$o$}
    \tkzLabelPoint[below](w){$x$}
     \tkzLabelPoint[below](w''){$x''$}
     \tkzLabelPoint[below](w'){$x'$}
       \tkzLabelPoint[left](z){$t$}
       \tkzLabelPoint[left](z'){$t'$}
       \tkzLabelPoint[left](z''){$t''$}
       \tkzLabelPoint[below](p){$p$}
       \tkzLabelPoint[right](q){$q$}
        \tkzLabelPoint[right](r){$r$}
        \tkzLabelPoint[above](v'){$v'$}
        \tkzLabelPoint[right](v){$v$}

    \end{tikzpicture}
    \caption{Lemma 12}
\end{figure}
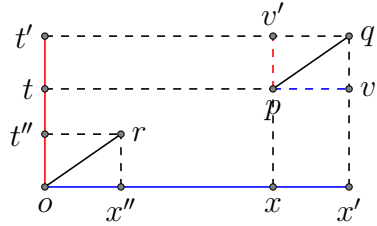

\subsubsection{Lemmata on uniqueness}

This concludes the preliminary results needed to prove the existence of a coordinatization. The proof that two coordinatizations $f$ and $f'$ are equivalent up to a rescaling and a Poincaré transformation will follow from \textit{Theorem 3.}, if we manage to show that there exists a rescaling U such that $f$ and $f'\circ U$ agree on the relativistic interval between any two points. We prove first that they agree on a basis. We then show in a sequence of steps that, if $f$ and $U \circ f'$ agree on a basis, then they must agree on the whole of $\mathfrak{M}$.

\begin{Lemma}
For all $p$, $q$, $r$ $\epsilon$ $\mathbb{R}^{4}$, if distance$_{4}$($p,r$) = distance$_{4}$($p,q$) $+$ distance$_{4}$($q,r$), then $\eta_{ab} (p,r) = \eta_{ab} (p,q) + \eta_{ab} (q,r) $.
\end{Lemma}

\begin{proof}
See [Suppes, 1959, p. 294]. 
\end{proof}

\begin{Lemma}
If $f$: U($\mathfrak{M}) \rightarrow \mathbb{R}^{4}$ satisfies conditions \textit{(1)-(2)} in \textbf{Theorem 4} and $Opp^\mathfrak{M}(o,t',o,x)$, then $\eta_{ab}(f(o),f(t')$) = - $\eta_{ab}(f(o),f(x)$).  
\end{Lemma}

\begin{proof}
$f$ and $f'$ satisfy the conditions of Lemma 6. 
\end{proof}

\begin{Lemma}
If $f$: U($\mathfrak{M}) \rightarrow \mathbb{R}^{4}$ and $f'$: U($\mathfrak{M}) \rightarrow \mathbb{R}^{4}$ are bijections satisfying conditions \textit{(1)-(2)} in \textbf{Theorem 4} and $f$ and $f'$ agree on the relativistic interval between two points $p$ and $q$ i.e., $\eta_{ab}(f(p),f(q)$) = $\eta_{ab}(f'(p),f'(q))$, then they agree on all the points that are collinear to $p$ and $q$.      
\end{Lemma}

\begin{proof}
Lemma 13 and condition (1) of Theorem 4 imply that $\eta_{ab}(f(p),f(n \cdot q)$) = $n \cdot \eta_{ab}(f(p),f( q)$) and $\eta_{ab}(f(p),f(n \cdot q)$) = $n \cdot \eta_{ab}(f'(p),f(q))$ for integer multiples of the segment $\vec{pq}$. The same holds for submultiples $\frac{1}{n}$. The result extends by continuity to all multiples of the segment $\vec{pq}$.

\end{proof}

\begin{Lemma}
Let $f$: U($\mathfrak{M})\rightarrow \mathbb{R}^{4}$ and $f'$: U($\mathfrak{M}) \rightarrow \mathbb{R}^{4}$ like in Lemma 15. Suppose that $Reached^\mathfrak{M}(o,p,q,p,r)$ for some $o,p,q, r$ $\epsilon$ U($\mathfrak{M}$). If $f$ and $f'$ agree on the components, that is $\eta_{ab}(f(o),f(p)$) = $\eta_{ab}(f'(o),f'(p))$ and $\eta_{ab}(f(o),f(q)$) = $\eta_{ab}(f'(o),f'(q))$, then  $\eta_{ab}(f(o),f(r)$) = $\eta_{ab}(f'(o),f'(r))$.   
\end{Lemma}

\begin{proof}
This is proven by cases. If $L^\mathfrak{M}$(o,r), we have that:

\begin{align*}
\eta_{ab}(f(o),f(r)) =  \eta_{ab}(f(o),f(o))  & \qquad(\text{By Axiom (F4)})\\
= 0 & \qquad(\text{Obvious}) \\
= \eta_{ab}(f'(o),f'(r) & \qquad (\text{similarly}) \\
\end{align*}

In all other cases  Axiom (SUM1) (SUM2) and a form of Pythagora's theorem for spacelike vectors imply that there is a point $r'$ that lies on either $\ell_{oq}$ or $\ell_{op}$ and $\equiv^\mathfrak{M}(o,r,o,r')$. Condition (2) and Lemma 15 imply the result.   \newline 
\end{proof}

{\fontsize{13}{0} \selectfont \textbf{Proof of Theorem 4.  }} 
 \newline

\textbf{Existence:}

\begin{proof}[\unskip\nopunct]

Let $\mathfrak{M}$ be a model of $\mathcal{M}^2$ in which the ordering of points on a line is a continuum. By (A1) it has elements $o, x, y, z, t$ such that $Basis^\mathfrak{M}(o,x,y,z,t)$. By the affine axioms and the axioms (A3)(A4)(DENS)(F5)(F6), it follows that the structure $E$ with congruence and betweenness restricted to the elements $v$ such that $Gen_{3D}^\mathfrak{M}(o,x,y,z,v)$ is a model of $\mathcal{E}_3$. By Theorem 1 it has a coordinatization $f$. An analogous statement is true for the structure on the line $L$ with the same relations restricted to points $l$ such that $Coll^\mathfrak{M}(o,t,l)$. By Theorem 2 it admits of a coordinatization $g$. Fix a total coordinatization $f'$ as specified in Definition 8. Axiom (A9) and the bijectivity of $f$ and $g$ imply that $f'$ is onto. Axiom (A2) implies that it is one-to-one.\newline

Let us fix a point $e$ $\epsilon$ $E$ such that $\eta_{ab}(f'(o),f'(e)) = 1$. \newline

Condition (1) is equivalent to the condition that, for all $p$, $q$, $r$ $\in$ $U(\mathfrak{M})$, $Bet^\mathfrak{M}(p,q,r)$ if and only if there is a positive constant $\lambda$ $\in$ $\mathbb{R}$ such  that $f'(r)-f'(p) = \lambda \cdot (f'(q)-f'(p))$. An analysis of (A8) and an appeal to Lemma 8 are sufficient to verify that this is the case. Let us now turn to (2).\newline  
Choose four points $p$, $q$, $l$, $r$ $\in$ $U(\mathfrak{M})$ such that $\equiv^\mathfrak{M}(l,r,p,q)$. By Lemma 12, two of the points can be chosen to be the origin $o$ (= $l$ = $p$). The proof that the relativistic interval, as computed by $f'$, is the same on the two segments proceeds by cases. (\textit{Case 1}) $L^\mathfrak{M}(o,q)$ implies $L^\mathfrak{M}(o,r)$ via (F1)(F2)(F3). Lemma 6 implies that $\eta_{ab}(f'(o),f'(r))$ = 0 = $\eta_{ab}(f'(o),f'(q)$. (\textit{Case 2}) If $S^\mathfrak{M}(o,q)$ we can assume by (SUM 1) that $r$ $\in$ $E$ and that there exist points $w$, $x$, $v_5$ $\in$ $E$ and $t$ $\in$ $L$ such that $Reached^\mathfrak{M}(o,x,t,x,q)$, $Bet^\mathfrak{M}(o,w,x)$,  $Opp^\mathfrak{M}(o,w,o,t)$, $Remterm^\mathfrak{M}(o,e, o,x,w,v_5)$ and $\equiv^\mathfrak{M}(o, p, x, v_5) )$.\newline
      
Lemma 6 implies equation (1). Lemma 13 and the fact that $Bet^\mathfrak{M}(o,w,x)$ justify equation (2) below. Lemma 9 on the calculus of segments implies equation (3). Equation (4) is from the definition of the interval $\eta_{ab}$.    

\begin{enumerate}[start=1,label={(\bfseries \arabic*)}]
\item $\eta_{ab}^2(f'(o),f'(w)) = - \eta_{ab}^2(f'(o),f'(t)) $
\end{enumerate}

\begin{enumerate}[start=2,label={(\bfseries \arabic*)}]

\item $\eta_{ab}^2(f'(o),f'(x)) =  \eta_{ab}^2(f'(o),f'(w)) + \eta_{ab}^2(f'(w),f'(x)) + 2\eta_{ab}(f'(o),f'(w)) \eta_{ab}(f'(w),f'(x)) $

\item  $\eta_{ab}^2(f'(o),f'(v_5)) =  \eta_{ab}^2(f'(w),f'(x)) +  2\eta_{ab}(f'(o),f'(w)) \eta_{ab}(f'(w),f'(x)) $
\end{enumerate}

\begin{enumerate}[start=4,label={(\bfseries \arabic*)}]

\item $\eta_{ab}^2(f'(o),f'(q)) =  \eta_{ab}^2(f'(o),f'(t)) +  \eta_{ab}^2(f'(o),f'(x)) $ \newline

\end{enumerate}

By substituting in (4) the two terms for their equivalents in (2) and in (1) leaves the expression that figures on the right-hand side of (3). This proves that $\eta_{ab}^2(f'(o),f'(q))$ = $\eta_{ab}^2(f'(o),f'(v_5))$. The transitivity of congruence implies $\equiv^\mathfrak{M}(o, r, x, v_5) )$. Theorem 1, and the fact that $o$, $r$, $v_5$ $\in$ $E$, imply that also $\eta_{ab}^2(f'(o),f'(r))$ = $\eta_{ab}^2(f'(o),f'(v_5))$. (\textit{Case 3}) when $T^\mathfrak{M}(o,q)$ is analogous. \newline
To prove the conditionals in the other direction, it suffices to note \textit{(Case 1)} that lightlike segments are congruent by (F4). By (4), lemma 6 and (SUM3) it follows that, if $\eta_{ab}^2(f'(o),f'(q)) = 0 = \eta_{ab}^2(f'(o),f'(r))$, then $L^\mathfrak{M}(o,r)$ and   $L^\mathfrak{M}(o,q)$. \textit{(Case 2)} and \textit{(Case 3)} follow from Theorem 1 and (SUM1) (SUM2) and the transitivity of congruence. The biconditionals in (2) of Theorem 4 in terms of `$<_{\equiv}$' follow readily from the biconditionals in terms of `$\equiv$'. We have already noted the fact that segments lightlike in $\mathfrak{M}$ have null interval relative to $f'$. Timelike and spacelike segments are congruent to segments in $E$ and $L$ respectively by (SUM1)(SUM2) and the square of the interval is respectively negative and positive between points in $E$ and $L$ by the construction of $f'$.\footnote{Within the three main categories the relation `$<_{\equiv}$' is definable in terms of `$\equiv$'.}\newline

\textbf{Uniquess up to a rescaling and a Poincaré transformation:} \newline

Let $f: U(\mathfrak{M}) \rightarrow \mathbb{R}^{4}$ and $f': U(\mathfrak{M}) \rightarrow \mathbb{R}^{4}$ be two bijective functions that satisfy conditions (1) and (2) of \textit{Theorem 4}. Fix the basis $o$, $x$, $y$, $z$ and $t$ that is associated by $f$ to the canonical basis of $\mathbb{R}^{4}$. There is a real $k$ such that $\eta_{ab}(f'(o),f'(x)) = k$ and (by Lemma 6) $\eta_{ab}^2(f(o),f(t))= -k$. Let $U: \mathbb{R}^{4} \rightarrow \mathbb{R}^{4}$ be the rescaling function such that $U(\vec{x})$ = $\frac{1}{k} \cdot \vec{x}$. It will suffice to show, by Theorem 3, that $U \circ f'= f''$ and $f$  agree on the relativistic interval between all $p$ and $q$ in $U(\mathfrak{M})$. By condition (2) and Lemma 6 it will suffice check the case when $p=o$. $U \circ f'=f''$ and $f$ agree on the basis $o$, $x$, $y$, $z$ and $t$ by construction. By axiom (A2) we have that $Gen_{4D}^\mathfrak{M}(o,x,y,z,t,q)$. By analysing the definition and noticing axiom (A7) we get a sequence $x'$, $y'$, $z'$ such that $Reached^\mathfrak{M}(o,x,y,x',y')$, and $Reached^\mathfrak{M}(o, y', z, y', z')$ and finally that $Reached^\mathfrak{M}(o, z', t, z', q)$ . The first equality $\eta_{ab}(f(o),f(x')$) = $\eta_{ab}(f''(o),f''(x'))$ follows from Lemma 15. The fact that $\eta_{ab}(f(o),f(y')$) = $\eta_{ab}(f''(o),f''(y'))$ and ultimately that $\eta_{ab}(f(o),f(x')$) = $\eta_{ab}(f''(o),f''(q))$ follows by successive applications of Lemma 16.   

\end{proof}

\section{Conclusion and future directions}

We have proposed a formalization of a small fragment of physical theory for a specific purpose, but let us conclude with some other uses that it might serve. We see two main directions in which this type of work can lead. One is to attempt to regiment more complex physical theories. We may begin by adding a classical field to our empty spacetime and formalize something like relativistic electrodynamics. To keep the axiom system intrinsic, it is preferable to avoid simply stating an analog of some differential equations - like Maxwell Equations - under some foliation of spacetime. It is preferable to develop a part of integration theory and use the theorems of multivariable calculus to rephrase the laws without coordinates -  a task for a future article. Other classical gauge theories can be dealt with in the same way. This may be done by introducing a six place mixed predicate for each scalar field, so as to compare the ratio of the intensity of the field at two points with the ratio in the length between two segments. A second natural step forward requires us to move away from a flat spacetime and to attempt to describe axiomatically the geometry of curved Lorentzian manifolds. Our present work should prove again useful: since a manifold is something that has, at each point a minkowskian tangent plane, this means in nominalistic terms that it approximates our axioms on small patches.   Whether manifolds of this sort can be treated with a predicate of betweenness on local geodesics and comparative predicates for proper time along paths deserves investigation. \newline   
This work pushes in the same direction as the program of nominalization of [Field 1980]. With each step, we augment the amount of nominalistic physics at our disposal. It is in general useful and illuminating to proceed further while trying to introduce as little further apparatus as possible; even when the apparatus is nominalistically acceptable. It is an interesting question how much of differential geometry or physics, for example, can be formalized without resorting to the calculus of individuals or mereology. Quantification over regions, regular curves and aggregates of points appears, at first sight, to be needed to describe the trajectory of a particle when that trajectory is not inertial. But how far we can go without mereology remains an open question [Martin 1971]. The introduction of mereology marks marks a crucial transition. For its addition to our geometric theory turns a decidable theory into an undecidable theory in the gödelian sense. This brings us to a second goal: to study better the metatheory of the system. We have not attempted to check whether it admits quantifier elimination upon the addition of primitives, or whether the theory of o-minimality can in some sense be applied to our geometric theory (see [Van den Dries 1998]). It is not clear to us whether this is a fruitful terrain for much model thoery. But these two lines of enquiry - the formalization of more theories and the study of their model theory by using higher mathematical logic - do not pull in opposite directions. They ought to proceed hand in hand. What will be achieved by these combined efforts remains to be seen.

\end{document}